\documentclass[sigconf, nonacm]{acmart}

\usepackage{ifthen}
\usepackage{alltt}
\usepackage{enumitem}
\usepackage{listings}
\usepackage{mathtools}

\provideboolean{ISARXIV}  % is this for arXiv? default: true
\setboolean{ISARXIV}{true}
\newcommand{\ifVldbElse}[2]{\ifthenelse{\not\boolean{ISARXIV}}{#1}{#2}}

\newtheorem{claim}{Claim}
\newcommand\vldbdoi{XX.XX/XXX.XX}
\newcommand\vldbpages{XXX-XXX}
\newcommand\vldbvolume{16}
\newcommand\vldbissue{11}
\newcommand\vldbyear{2023}
\newcommand\vldbauthors{\authors}
\newcommand\vldbtitle{\shorttitle} 
\newcommand\vldbavailabilityurl{https://github.com/IBM/sliding-window-aggregators}
\newcommand\vldbpagestyle{empty} 

\newcommand*{\onef}[0]{\mathbf{1}}
\newcommand*{\pseudocode}[1]{\texttt{\small #1}}
\newcommand*{\tv}[2]{\big[{#1\atop #2}\big]}

\newtheorem{theorem}{Theorem} % do not use section, it fails in the appendix

\lstdefinelanguage{MyPseudoCode}{
  columns=flexible,
  keepspaces=true,
  basicstyle={\small\ttfamily},
  morekeywords={and,break,each,else,false,for,fun,if,in,not,null,or,return,true,while},
  keywordstyle=\bfseries,
  identifierstyle=\slshape,
  numbers=left,
  numberstyle=\sffamily\tiny,
  numbersep=2mm,
  xleftmargin=4mm,
  mathescape=true,
  escapeinside={(*@}{@*)},
}

\begin{document}

\ifVldbElse{%
\title[Out-of-Order Sliding-Window Aggregation with Efficient Bulk Evictions and Insertions]
{Out-of-Order Sliding-Window Aggregation with\\Efficient Bulk Evictions and Insertions}
}{%
\title[Out-of-Order Sliding-Window Aggregation with Efficient Bulk Evictions and Insertions (Extended Version)]
{Out-of-Order Sliding-Window Aggregation with\\[-1pt]Efficient Bulk Evictions and Insertions (Extended Version)\vspace*{-1pt}}
}

%%
%% The "author" command and its associated commands are used to define the authors and their affiliations.
\author{Kanat Tangwongsan}
\affiliation{%
  %\institution{\mbox{\hspace*{-4mm}Mahidol University International College}}
  \institution{\mbox{Mahidol University International College}}
}
\email{kanat.tan@mahidol.edu}

\author{Martin Hirzel}
\affiliation{%
  \institution{IBM Research}
}
\email{hirzel@us.ibm.com}

\author{Scott Schneider}
\affiliation{%
  \institution{Meta}
}
\email{scottas@meta.com}

\begin{abstract}
Sliding-window aggregation is a foundational stream processing
primitive that efficiently summarizes recent data.
The state-of-the-art algorithms for sliding-window aggregation are
highly efficient when stream data items are evicted or inserted one at
a time, even when some of the insertions occur out-of-order.
However, real-world streams are often not only out-of-order but also
burtsy, causing data items to be evicted or inserted in larger bulks.
This paper introduces a new algorithm for sliding-window aggregation
with bulk eviction and bulk insertion.
For the special case of single insert and evict, our algorithm matches
the theoretical complexity of the best previous out-of-order algorithms.
For the case of bulk evict, our algorithm improves upon the
theoretical complexity of the best previous algorithm for that case
and also outperforms it in practice.
For the case of bulk insert, there are no prior algorithms, and our
algorithm improves upon the naive approach of emulating bulk insert
with a loop over single inserts, both in theory and in practice.
Overall, this paper makes high-performance algorithms for sliding
window aggregation more broadly applicable by efficiently handling the
ubiquitous cases of out-of-order data and bursts.

\end{abstract}

\maketitle

\ifVldbElse{%
\pagestyle{\vldbpagestyle}
\begingroup\small\noindent\raggedright\textbf{PVLDB Reference Format:}\\
\vldbauthors. \vldbtitle. PVLDB, \vldbvolume(\vldbissue): \vldbpages, \vldbyear.\\
\href{https://doi.org/\vldbdoi}{doi:\vldbdoi}
\endgroup
}{%
\pagestyle{plain}
\vspace*{10mm}
}
\begingroup
\renewcommand\thefootnote{}\footnote{\noindent
\ifVldbElse{%
This work is licensed under the Creative Commons BY-NC-ND 4.0 International License. Visit \url{https://creativecommons.org/licenses/by-nc-nd/4.0/} to view a copy of this license. For any use beyond those covered by this license, obtain permission by emailing \href{mailto:info@vldb.org}{info@vldb.org}. Copyright is held by the owner/author(s). Publication rights licensed to the VLDB Endowment. \\
\raggedright Proceedings of the VLDB Endowment, Vol. \vldbvolume, No. \vldbissue\ %
ISSN 2150-8097. \\
\href{https://doi.org/\vldbdoi}{doi:\vldbdoi} \\
}{%
This paper is an extended version of our VLDB 2023 paper ``Out-of-Order Sliding-Window Aggregation with Efficient Bulk Evictions and Insertions''. It adds an appendix with proofs, pseudocode, and examples that did not fit in the page limit.
\vspace*{12mm}
}
}\addtocounter{footnote}{-1}\endgroup

\ifdefempty{\vldbavailabilityurl}{}{
\vspace{.3cm}
\begingroup\small\noindent\raggedright\textbf{PVLDB Artifact Availability:}\\
The source code, data, and/or other artifacts have been made available at \url{\vldbavailabilityurl}.
\endgroup
}

\section{Introduction}\label{sec:introduction}

In data stream processing, a sliding window covers the most recent
data, and sliding-window aggregation maintains a summary of it.
Sliding-window aggregation is a foundational primitive for stream
processing, and as such, is both widely used and widely supported.
In various application domains, stream processing must have low
latency; for example, late results can cause financial losses in
trading or harm property and lives in security or transportation.
Furthermore, streaming data often arrives out-of-order, but new data
items must be incorporated into a sliding window at their correct
timestamps and the aggregation may not be commutative.
Finally, data streams do not always have a smooth rate: in the real world,
data items often enter and depart sliding windows in bursts.

When streaming data is bursty, sliding-window aggregation needs to
support efficient bulk evictions and insertions to keep latency low.
In other words, it needs to evict or insert a bulk of $m$ data items
faster than it would take to evict or insert them one by one, lest
it incur a latency spike of $m\times$ that of a single operation.
Bulk evictions are common in time-based windows, where the arrival of
one data item at the youngest end of the window can trigger the
eviction of several data items at the oldest end.
For example, consider a window of size 60 seconds, with data items at
timestamps \mbox{\it [0.1,0.2,0.3,0.4,0.5,10,20,30,40,50,60]} seconds.
If the next data item to be inserted has timestamp \textit{61}, the window
must evict the items at timestamps \mbox{\it [0.1,0.2,0.3,0.4,0.5]}.
Since these are $m=5$ items, evicting them one
by one would incur a $5\times$ latency spike.

While a small bulk (e.g., $m=5$) is harmless, bursts can
result in $m$ in the thousands of data items or more.
For instance, data streams may experience transient outages, causing
bursts during recovery~\cite{bouillet_et_al_2012}.
Besides time-based windows, applications may use other window types
such as sessions~\cite{traub_et_al_2019} or data-driven adaptive
windows~\cite{bifet_gavalda_2007}.
Streaming systems may internally use implementation techniques that
introduce disorder~\cite{li_et_al_2008}.
When a streaming system receives multiple streams from different
data sources, their logical times may drift against each
other~\cite{krishnamurthy_et_al_2010}.
Real-world events, such as breaking news, severe weather, rush hour traffic,
sales, accidents, opening of stores or stock markets,
etc.~can cause bursty streams~\cite{poppe_et_al_2021}.
All these scenarios 
necessitate sliding-window aggregation with efficient bulk evictions and
insertions---without harming the tuple-at-a-time performance.  

The literature has few solutions to this problem, and none match
our solution in completeness or algorithmic complexity.
List-based approaches such as Two-Stacks handle
  neither out-of-order nor bulk
  operations~\cite{tangwongsan_hirzel_schneider_2021}.
The AMTA algorithm only handles in-order windows and only offers bulk
eviction but not bulk insertion~\cite{villalba_berral_carrera_2019}.
CPiX has a linear factor in its algorithmic complexity
for bulk eviction and is limited to commutative aggregation over
time-based windows~\cite{bou_kitagawa_amagasa_2021}.
The FiBA algorithm is optimal for out-of-order sliding-window
aggregation with single evictions and insertions but does not directly
support bulk operations~\cite{tangwongsan_hirzel_schneider_2019}.
While the literature on balanced tree algorithms provides partial
solutions to bulk evictions and
insertions~\cite{brown_tarjan_1979,kaplan_tarjan_1995,hinze_paterson_2006},
each paper solves a different part of the problem using a different
data structure, and none offer incremental aggregation.
Section~\ref{sec:relatedwork} discusses related work in more detail.

Our new solution builds on FiBA~\cite{tangwongsan_hirzel_schneider_2019}, 
a B-tree augmented with fingers and with location-sensitive
partial aggregates.
The fingers help efficiently find tree nodes to be manipulated
when the window slides. The location-sensitive partial aggregates
avoid propagating local updates to the root in most cases.
Intuitively, our bulk eviction and insertion  
have three steps:
\begin{itemize}[topsep=10pt]
  \item a finger-based \emph{search} to find the affected nodes of the
    tree;
  \item a single shared \emph{pass up} the tree to insert or evict
    items in bulk while also repairing any imbalances this causes;
    and
  \item a single shared \emph{pass down} the affected spine(s) of the
    tree to repair location-sensitive partial aggregates stored there.
\end{itemize}
The trick for efficient bulk evict is to not look at each evicted entry
individually, but rather, only cut the tree along the boundary between
the entries that go and those that stay.
The trick for efficient bulk insert is to share work caused by
multiple inserted entries as low down in the tree as possible, i.e., to
process paths from insertion sites together as soon as they converge.

Let $n$ be the window size
(the number of data items currently in the window);
$m$, the bulk size
(the number of data items being evicted or inserted); and
$d$, the out-of-order insertion distance (the number of data items
in the part of the window that overlaps with the bulk).
Our algorithm performs bulk eviction in amortized $O(\log m)$
time and bulk insertion in amortized $O(m\log\frac{d}{m})$ time.
Neither of these two time bounds depend on the window size~$n$
and bulk eviction is sublinear in the bulk size~$m$.
For $m=1$, the amortized time matches the proven lower
bounds of $O(1)$ for eviction and $O(\log d)$ for out-of-order
insertion, which means $O(1)$ for in-order insertion at the smallest~$d$.
The worst-case time complexity is $O(\log n)$ for bulk evict and
$O(m\log(\frac{m+n}{m})+\log d)$ for bulk insert, because the pass up
the tree can reach the root in the worst case.
This worst case is guaranteed to be so rare that in the
long run, the amortized complexity prevails.
It uses $O(n)$ space, with the constant
depending on the B-tree's arity.

We implemented our algorithm in C++ and made it available at
\url{https://github.com/IBM/sliding-window-aggregators}, along with our
implementations of other sliding-window aggregation algorithms we compare with
experimentally. Commit \verb@f3beed2@ was used in the experiments of this
paper. Our experimental results demonstrate that our bulk evict yields the best
latency compared to several state-of-the-art baselines, and our
bulk insert yields the best latency for the out-of-order case (which most
algorithms do not support at all in the first place). Overall, this paper
presents the first algorithm for efficient bulk insertions in sliding windows,
and the algorithm with the best time complexity so far for bulk evictions from
sliding windows.
% It helps keep latency spikes down when aggregating bursty streams.

\section{Related Work}\label{sec:relatedwork}

Before our work, the most efficient algorithm for in-order sliding window
aggregation with bulk eviction was AMTA~\cite{villalba_berral_carrera_2019}.
AMTA supports single inserts or evicts in amortized $O(1)$ time.
Given a window of size $n$, it supports bulk evict in amortized
$O(\log n)$ time.
However, AMTA does not directly support bulk insertion, so inserting
$m$ items
% with a loop of $m$ single inserts
takes amortized $O(m)$ time.
Our algorithm matches AMTA's amortized complexity for single inserts
and evicts, and improves bulk evict to amortized $O(\log m)$ time.
Unlike our algorithm, AMTA does not support out-of-order insert.

CPiX supports both bulk eviction and bulk insertion, including
out-of-order insertion~\cite{bou_kitagawa_amagasa_2021}.
The paper states the time complexity of bulk insert or evict as
\mbox{$(p_1 + 1) \log(|\frac{n}{k}|) + 3p_2$},
where the number $k$ of checkpoints is recommended to be~$\sqrt{n}$;
$p_1$~is the number of affected partitions in the oldest checkpoint; and
$p_2$~is the number of affected partitions in the remaining checkpoints.
Given \mbox{$O(\log(|\frac{n}{\sqrt{n}}|)) = O(\log n)$},
and assuming $p_1$ and $p_2$ are proportional to the batch size $m$,
this corresponds to an amortized time of $O(m \log n)$.
This is worse than AMTA's $O(\log n)$ and our $O(\log m)$ for bulk evict.
Moreover, 
unlike our algorithm, CPiX only works for time-based windows and
commutative aggregation.

The most efficient prior algorithm for out-of-order sliding window
aggregation is FiBA~\cite{tangwongsan_hirzel_schneider_2019}.
It supports a single insert or evict in amortized $O(\log d)$ time,
where $d$ is the distance of the operation from either end of the window.
% In the in-order case, this means amortized $O(1)$ time.
FiBA can emulate bulk insert or evict using loops of $m$ single
inserts or evicts for a time complexity of $O(m \log d)$.
Our new algorithm improves upon this baseline.
% with dedicated support for bulk evict and out-of-order bulk insert.

Some streaming systems limit out-of-order distance to a watermark
\cite{akidau_et_al_2013}; instead, our algorithm implements the more general
case that requires no such a priori bounds.

Our algorithm is inspired by the literature on bulk operations for
balanced trees.
Brown and Tarjan show how to merge two height-balanced trees of sizes $m$
and $n$, where $m<n$, in $O(m\log\frac{n}{m})$ steps~\cite{brown_tarjan_1979}.
The keys of the two trees can be interspersed, so their algorithm
corresponds to our out-of-order bulk insertion scenario.
Unlike our algorithm, theirs supports neither aggregation
nor bulk eviction.
Furthermore, our algorithm improves the complexity to
$O(m\log\frac{d}{m})$, where $d$ is the overlap between the two trees.
Kaplan and Tarjan show how to catenate two height-balanced trees
in worst-case $O(1)$ time~\cite{kaplan_tarjan_1995}.
But they do not allow keys to be interspersed, so their
approach is restricted to the in-order case.
Also, unlike our algorithm, their approach does not perform
aggregation and does not support bulk eviction.
Hinze and Paterson show how to both split and merge balanced trees
in amortized $O(\log d)$ time~\cite{hinze_paterson_2006}.
However, their merge does not allow keys to be interspersed,
so it corresponds to in-order bulk insertion.
Also, their approach does not perform sliding window aggregation.

The sliding-window aggregation literature also pursues other
objectives besides bulk eviction and
out-of-order bulk insertion. %, which are the focus of this paper.
Scotty optimizes for coarse-grained sliding, performing pre-aggregation
to take advantage of co-eviction~\cite{traub_et_al_2019}.
Their work shows how to handle all combinations of order, window
kinds, aggregation operations, etc., and is complementary to this paper.
  ChronicleDB uses a temporal aggregate B+-tree and optimizes writes to
  persistent storage while handling moderate amounts of out-of-order data by
  leaving some free space in each block~\cite{seidemann_et_al_2019}. Hammer
  Slide uses SIMD instructions to speed up sliding-window
  aggregation~\cite{theodorakis_et_al_2018}; SlideSide generalizes it to the
  multi-query case~\cite{theodorakis_pietzuch_pirk_2020}; and LightSaber
  further generalizes it for parallelism~\cite{theodorakis_et_al_2020}.
DABA Lite performs both single in-order insert and single evict in
worst-case $O(1)$ time but does not support out-of-order
insert~\cite{tangwongsan_hirzel_schneider_2021}.
FlatFIT focuses on window sharing for the in-order case, with
amortized $O(1)$ time for single insert and single evict, but does not
support out-of-order insert~\cite{shein_chrysanthis_labrinidis_2017}.
None of the above directly support bulk operations;
they can do $m$ inserts or evicts using simple loops,
with an algorithmic complexity of $m$ times
that of their single-operation complexity.

\section{Background}\label{sec:background}

This section formalizes the problem solved in this paper and reviews
known concepts such as monoids and finger B-trees upon which our work
builds.

\subsection{Problem Statement}\label{sec:problem}

\paragraph{Monoids.}
A monoid is a triple $(S, \otimes, \onef)$ with a set $S$, an
associative binary combine operator $\otimes$, and a neutral
element~$\onef$.
Several common aggregation operators are monoids, including count,
sum, min, and max.
Furthermore, several more common aggregation operators can be lifted
into monoids, including arithmetic or geometric mean, standard
deviation, argMax, maxCount, first, last, etc.
Even several sophisticated statistical and machine learning operators
can be lifted into monoids, including mergeable
sketches~\cite{agarwal_et_al_2012} such as Bloom filters or algebraic
classifiers~\cite{izbicki_2013}.
Associativity means that
\mbox{$\forall v_1,v_2,v_3\in S:(v_1\otimes v_2)\otimes v_3=v_1\otimes(v_2\otimes v_3)$}.
That means we can omit the parentheses and simply write
\mbox{$v_1\otimes v_2\otimes v_3$}; furthermore, in this paper, we sometimes
even omit the $\otimes$ and simply use $v_1v_2v_3$ product notation.
By giving flexibility over how values are grouped during combining,
associativity is essential to most incremental sliding-window
aggregation algorithms.
The identity element $\onef$ satisfies
\mbox{$\forall v\in S:\onef\otimes v=v=v\otimes\onef$}.
It gives meaning to aggregation over empty (sub)windows.
For a monoid, while the combine operator $\otimes$ must be
associative, it does not need to be invertible or commutative.
Thus, any sliding-window aggregation algorithm that works for general
monoids must handle the case of non-invertible
and non-commutative operators.

\paragraph{Abstract Data Type.}
Below we define an abstract data type with three operations
\pseudocode{query}, \pseudocode{bulkEvict}, and
\pseudocode{bulkInsert}.
Our formulation decouples these three operations to make them as
versatile as possible, so they can be used in any order, with any kind
of window specification, including windows that grow and shrink
dynamically.
Of course, an abstract data type is not itself an algorithm; instead,
it merely specifies the behavior that a given algorithm should
implement.
While it is easy to implement the abstract data type with a
brute-force algorithm that recomputes everything from scratch,
our problem statement is to design, analyze, and evaluate an incremental
algorithm\footnote{An \emph{incremental algorithm} keeps partial results to
avoid from-scratch recomputations where possible.} with native support for bulk
operations that have better asymptotic and practical time complexity than
before.

\paragraph{Query.}
The operation \pseudocode{query()} makes no changes to the window
and computes the monoidal combination of all values currently in the
window in the order of their timestamps.
Let the window contents be
\mbox{$W=\tv{t_1}{v_1},\ldots,\tv{t_n}{v_n}$} where $t_i<t_{i+1}$.
Then \pseudocode{query()} returns \mbox{$v_1\otimes\ldots\otimes v_n$},
or the neutral element $\onef$ if the window is empty.

\paragraph{Bulk Eviction.}
The operation \mbox{\pseudocode{bulkEvict(}$t$\pseudocode{)}} removes all
entries with timestamps $\le t$ from the window, leaving the entries
with timestamps~$>t$.
Let the window contain
\mbox{$W^\textrm{pre}=\left\{\tv{t_1^\textrm{pre}}{v_1^\textrm{pre}},\ldots,\tv{t_n^\textrm{pre}}{v_n^\textrm{pre}}\right\}$} before eviction.
Then, the window contents post-eviction are
\[W^\textrm{post}=\left\{\tv{t^\textrm{post}}{v^\textrm{post}}:\tv{t^\textrm{post}}{v^\textrm{post}}\in W^\textrm{pre}\wedge t^\textrm{post}>t\right\}.\]

\paragraph{Bulk Insertion.}
The operation
\mbox{\pseudocode{bulkInsert(}$B^\textrm{in}$\pseudocode{)}},
where the contents of the bulk to be inserted are
\mbox{$B^\textrm{in}=\left\{\tv{t_1^\textrm{in}}{v_1^\textrm{in}},\ldots,\tv{t_m^\textrm{in}}{v_m^\textrm{in}}\right\}$}
with \mbox{$t_i^\textrm{in}<t_{i+1}^\textrm{in}$},
interleaves the previous window contents with the bulk in temporal
order, and in case of collisions, combines them.
Let the window contents pre-insertion be
\mbox{$W^\textrm{pre}=\left\{\tv{t_1^\textrm{pre}}{v_1^\textrm{pre}},\ldots,\tv{t_n^\textrm{pre}}{v_n^\textrm{pre}}\right\}$}.
Then, the window contents post-insertion are
\[\begin{array}{r@{\,}c@{\,}l}
W^\textrm{post}
& = & \left\{\tv{t^\textrm{post}}{v^\textrm{post}}:\tv{t^\textrm{post}}{v^\textrm{post}}\in W^\textrm{pre}\wedge\not\exists v^\textrm{in}:\tv{t^\textrm{post}}{v^\textrm{in}}\in B^\textrm{in}\right\}\\
& \cup & \left\{\tv{t^\textrm{post}}{v^\textrm{post}}:\tv{t^\textrm{post}}{v^\textrm{post}}\in B^\textrm{in}\wedge\not\exists v^\textrm{pre}:\tv{t^\textrm{post}}{v^\textrm{pre}}\in W^\textrm{pre}\right\}\\
& \cup & \left\{\tv{t^\textrm{post}}{v^\textrm{pre}\otimes v^\textrm{in}}:\tv{t^\textrm{post}}{v^\textrm{pre}}\in W^\textrm{pre}\wedge\tv{t^\textrm{post}}{v^\textrm{in}}\in B^\textrm{in}\right\}.
\end{array}\]

\subsection{FiBA Data Structure}\label{sec:fiba}

\begin{figure}
\centerline{\includegraphics[width=\columnwidth]{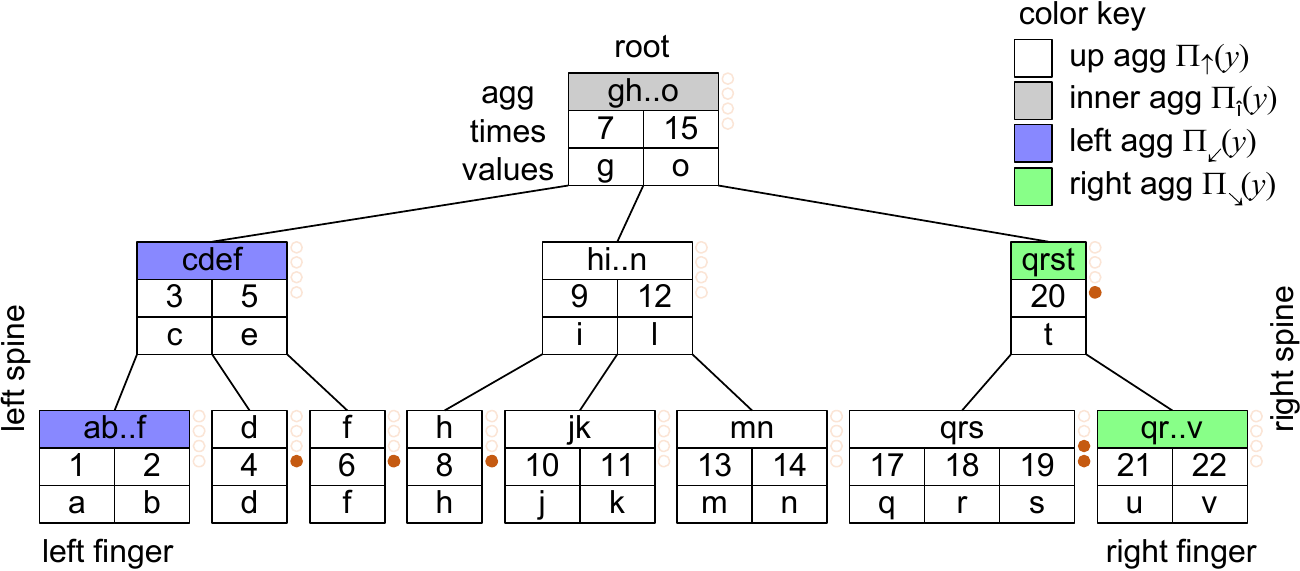}}
\caption{\label{fig:tree_with_aggregates}FiBA data structure example.}
\end{figure}

The FiBA data structure is a finger B-tree augmented for
sliding-window aggregation.
It was first introduced to optimize single out-of-order eviction and
insertion~\cite{tangwongsan_hirzel_schneider_2019}.
This paper uses the same data structure but introduces a new algorithm for bulk
eviction and bulk insertion operations.
Figure~\ref{fig:tree_with_aggregates} gives a running example.
%% for the rest of this section.

\paragraph{Node Contents.}
Each node stores a partial aggregate \pseudocode{agg} and two parallel
arrays of times and values $\tv{t_i}{v_i}$.
For example, the aggregate of the second leaf from the right in
Figure~\ref{fig:tree_with_aggregates} is $\pseudocode{agg}=qrs$.
Recall that $qrs$ is shorthand for \mbox{$q\otimes r\otimes s$}.
While the data structure works with any monoid, assume for
illustration that the monoid here is $\otimes=\max$ and that the
values are $q=4$, $r=2$, and $s=5$.
Then \mbox{$\pseudocode{agg}=q\otimes r\otimes s=\max(4,2,5)=5$}.
In this example, timestamps are integers with almost no gaps, but
there is a gap between times 15 and~17.
In general, %% timestamps need not be integers; indeed,
any totally ordered set will do for timestamps, and the
data structure allows any number of gaps of any sizes.

\paragraph{Location-Sensitive Partial Aggregates.}
At first glance, it would seem intuitive to set \pseudocode{agg} as the
aggregation of a node and all its children and descendants.
However, that would be suboptimal for sliding-window aggregation,
because it would require propagating all window changes to the root,
with time complexity $O(\log n)$.
For a better time complexity, FiBA stores one of four different
kinds of aggregate at each node depending on its location in the tree.
Below are definitions of those four kinds of aggregates:
up aggregate~$\Pi_\uparrow$, inner aggregate~$\Pi_{\hat{\scriptscriptstyle |}}$,
left aggregate~$\Pi_\swarrow$, and right aggregate~$\Pi_\searrow$.
In each of these definitions, let the current node under discussion be
$y$ with arity $a$, children $c_0,\ldots,c_{a-1}$, and values
$v_0,\ldots,v_{a-2}$.

\begin{itemize}[leftmargin=1.5em]%[leftmargin=0pt,itemindent=1em,label=$\rhd$,labelwidth=-1pt]
  \item The \emph{up aggregate} includes all children of $y$ and all
    of $y$'s own values in timestamp order:
\[\Pi_\uparrow(y)=\Pi_\uparrow(c_0)\otimes v_0\otimes\ldots\otimes v_{a-2}\otimes\Pi_\uparrow(c_{a-1})\]
    For example, the node with timestamps 9 and 12 in the middle of
    Figure~\ref{fig:tree_with_aggregates} has $\pseudocode{agg}=hi..n$,
    which is the ordered monoidal combination starting from its
    left-most child, through all values and children, up to and
    including its right-most child.
    For a more concrete example, assume $h=4$,
      $i=5$, $j=1$, $k=3$, $l=5$, $m=4$, $n=2$ and the
      \pseudocode{max} monoid, then \pseudocode{agg=5}.

  \item The \emph{inner aggregate} includes all of $y$'s own values
    and inner children but excludes the left-most and right-most child:
\[\Pi_{\hat{\scriptscriptstyle |}}(y)=v_0\otimes\Pi_\uparrow(c_1)\otimes\ldots\otimes\Pi_\uparrow(c_{a-2})\otimes v_{a-2}\]
    For example, the root in Figure~\ref{fig:tree_with_aggregates} has
    $\pseudocode{agg}=gh..o$, which combines its left value $g$ with the
    aggregate of only the middle child $hi..n$ and the right value~$o$.
    This means that the root stores an aggregate of the entire tree
    except for the left and right spines and their descendants.

  \item The \emph{left aggregate} excludes the leftmost child but
    includes all of $y$'s own values and the rightmost child, and then
    combines that with the parent $x$ (unless $x$ is the root):
\[\qquad\Pi_\swarrow(y)=\Pi_{\hat{\scriptscriptstyle |}}(y)\otimes\Pi_\uparrow(c_{a-1})\otimes\left(\begin{array}{ll}
  \onef & \textrm{if }x\textrm{ is }\textit{root}\\
  \Pi_\swarrow(x) & \textrm{otherwise}
\end{array}\right)\]
    For example, the left-most leaf of
    Figure~\ref{fig:tree_with_aggregates} has aggregate
    $\pseudocode{agg}=ab..f$, which combines its own values $ab$ with the
    aggregate $cdef$ of its parent, resulting in an aggregate of the
    entire left spine and all its descendants.

  \item The \emph{right aggregate} combines the aggregate of the
    parent $x$ (unless $x$ is the root) with all of $y$'s own values
    and most children but excludes the rightmost child:
\[\qquad\Pi_\searrow(y)=\left(\begin{array}{ll}
  \onef & \textrm{if }x\textrm{ is }\textit{root}\\
  \Pi_\searrow(x) & \textrm{otherwise}
\end{array}\right)\otimes\Pi_\uparrow(c_0)\otimes\Pi_{\hat{\scriptscriptstyle |}}(y)\]
    For example, the right-most leaf of
    Figure~\ref{fig:tree_with_aggregates} has aggregate
    $\pseudocode{agg}=qr..v$, which combines the aggregate $qrst$ of its
    parent with its own values $uv$, resulting in an aggregate of the
    entire right spine and all its descendants.
\end{itemize}

\paragraph{Representation.}
The tree is represented by three pointers:
\emph{left finger} to the left-most child;
the \emph{root};
and \emph{right finger} to the right-most child.
Each node stores its location-sensitive partial aggregate
\pseudocode{agg}, times, and values, and in addition, has pointers to its
parent and children, if any.
Finally, each node stores two Boolean flags to indicate whether it is
on the left or right spine, respectively.

\paragraph{Invariants.}
The following properties about height, order, arity, and aggregates
hold before each eviction or insertion and must be established again
by the end of each eviction or insertion.

The \emph{height invariant} requires all leaves to have the exact same
distance from the root.

The \emph{order invariant} says that the times
\mbox{$t_0,\ldots,t_{a-2}$} within each node are ordered, i.e.,
\mbox{$\forall i:t_i<t_{i+1}$};
and furthermore, if a node has children \mbox{$c_0,\ldots,c_{a-1}$},
then for all $i$, $t_i$ is greater than all times in $c_i$ or its
descendants and smaller than all times in $c_{i+1}$ or its
descendants.

The \emph{arity invariants} constrain the sizes of nodes to keep the
tree balanced.
Each node has an arity $a$, and different nodes can have different
arities.
For non-leaf nodes, $a$ is the number of children.
All nodes have $a-1$ entries, i.e., parallel arrays of $a-1$
timestamps and $a-1$ values.
There is a data structure hyperparameter \pseudocode{MIN\_ARITY},
which is an integer $>1$, and
\mbox{$\pseudocode{MAX\_ARITY}=2\cdot\pseudocode{MIN\_ARITY}$}.
They constrain the arity of all non-root nodes to
\mbox{$\pseudocode{MIN\_ARITY}\le a\le\pseudocode{MAX\_ARITY}$}.
And for the root, \mbox{$2\le a\le\pseudocode{MAX\_ARITY}$}.
For example, \mbox{$\pseudocode{MIN\_ARITY}$} is $2$ in
Figure~\ref{fig:tree_with_aggregates}, so all non-leaf nodes have
\mbox{$2\le a\le4$} children, and all nodes have \mbox{$1\le a-1\le3$}
timestamps and values.

The \emph{aggregates invariants} govern which nodes store which kind
of location-sensitive aggregates, color-coded in
Figure~\ref{fig:tree_with_aggregates}.
All non-spine, non-root nodes store the up aggregate.
Nodes that are on the left spine but not the root store the left aggregate.
Nodes that are on the right spine but not the root store the right aggregate.
And the root stores the inner aggregate.
This means that the aggregate of the entire tree is simply the
combination of the aggregates of the left finger, the root, and the
right finger.
In other words, we can implement \pseudocode{query()} in constant
time by returning
\[\Pi_\swarrow(\pseudocode{leftFinger})\otimes\Pi_{\hat{\scriptscriptstyle |}}(\pseudocode{root})\otimes\Pi_\searrow(\pseudocode{rightFinger})\]

\paragraph{Imaginary Coins.}
To help prove the amortized time complexity,
% the original FiBA paper~\cite{tangwongsan_hirzel_schneider_2019}
we pretend that each node stores imaginary coins.
Figure~\ref{fig:tree_with_aggregates} shows these as small copper
circles.
Nodes that are close to underflowing store one coin to pay for the
rebalancing work in case of underflow.
Nodes that are close to overflowing store two coins to pay for the
rebalancing work in case of overflow.
Then, the proofs for amortized time complexity show that for any
possible sequence of operations, the algorithm always stores up
enough coins in advance at each node before it has to perform
eventual actual rebalancing work.

\section{Bulk Eviction}\label{sec:bulkeviction}

As defined in Section~\ref{sec:problem},
\pseudocode{bulkEvict(}$t$\pseudocode{)} removes all entries with
timestamps $\le t$ from the window.
So our algorithm must discard nodes to the
left of~$t$, keep nodes to the right of~$t$, and for nodes that
straddle the boundary, locally evict all entries up to $t$ and repair
any violated invariants.
Our bulk eviction algorithm has three steps:
\begin{enumerate}[label={\textbf{Step~\arabic*}}, leftmargin=*]
  \item A finger-based \emph{eviction boundary search} that returns a
    list, called \pseudocode{boundary}, of triples
    \mbox{\pseudocode{(node, ancestor, neighbor)}}.
  \item A~\emph{pass up} the \pseudocode{boundary}, and beyond as
    needed to repair invariants, that does the actual evictions and
    most repairs.
  \item A~\emph{pass down} the left spine, and if needed also the
    right spine, that repairs any leftover invariant violations.
\end{enumerate}

\medskip
\noindent\textbf{Bulk eviction Step 1: Eviction boundary search.}
This step finds the boundary to enable any
subsequent rebalancing operations during Step~2 to be
constant-time at each level.
For rebalancing to be efficient, it cannot afford
to trigger any searches of its own, and must instead rely on all
required searching to have already been done upfront.
Whereas textbook algorithms for B-trees with single evictions (such
as~\cite{cormen_leiserson_rivest_1990}) can repair arity invariants by
rebalancing with a node's left or right sibling, bulk eviction leaves
no left sibling.
That means the only eligible neighbor to help in rebalancing is the
right one, and that may have a different parent and thus not be a sibling.
Furthermore, rebalancing requires the least common ancestor of the
node and its neighbor, and that might not be their parent.
Hence, the job of the finger-based search is to find a list of
\mbox{\pseudocode{(node, ancestor, neighbor)}} triples, one for each
relevant level of the tree.
The search first starts at the left or right finger, whichever is
closest to $t$, and walks up the corresponding spine to find the top
of the boundary, i.e., the lowest spine node whose descendants straddle~$t$.
Then, the search traverses down to the actual eviction point while
populating the \pseudocode{boundary} data structure.
This downward traversal always keeps at most two separate chains for
the \pseudocode{node} and its \pseudocode{neighbor}, and can thus
happen in a single loop over descending tree levels.
If the search finds an exact match for $t$ in the tree, it stops
early; otherwise, it continues to a leaf and stops there.

\medskip
\noindent\textbf{Bulk eviction Step 2: Pass up.}
This step of the algorithm does most of the work: it performs
the actual evictions, and along the way, it also repairs most of the
invariants that those evictions may have violated.
Recall from Section~\ref{sec:fiba} that there are
invariants about height, order, arity, and location-sensitive partial
aggregates.
The pass up never violates invariants about height or order.
It immediately restores arity invariants, using some novel rebalancing
techniques described below.
Regarding aggregate invariants, the pass up only repairs aggregates
that follow a strictly ascending direction (up aggregates
$\Pi_\uparrow$ and inner aggregates~$\Pi_{\hat{\scriptscriptstyle |}}$).
The pass up leaves aggregates that involve the parent (left aggregates
$\Pi_\swarrow$ and right aggregates~$\Pi_\searrow$) to the later pass
down to repair.

The pass up has two phases: an eviction loop up the
\pseudocode{boundary} returned by the search, followed by a repair
loop further up beyond the boundary as long as there is more to repair.
At each level, the eviction loop performs the local eviction, repairs
arity underflow, and repairs local up aggregates or inner aggregates.
At each level that still needs such repair, the repair loop repairs
arity underflow and repairs local up aggregates or inner aggregates.
Aggregate repair happens in constant time per level by simply
recomputing aggregates of surviving affected nodes after eviction and
rebalancing are done.
Arity repair, also known as rebalancing, either moves entries from the
neighbor to the node or merges the node into the neighbor, depending
on their respective arities.
Let \pseudocode{nodeDeficit} be
\mbox{$\pseudocode{MIN\_ARITY}-\pseudocode{node.arity}$} and
let \pseudocode{neighborSurplus} be
\mbox{$\pseudocode{neighbor.arity}-\pseudocode{MIN\_ARITY}$}.
If \mbox{$\pseudocode{nodeDeficit}\le\pseudocode{neighborSurplus}$},
rebalancing does a move; otherwise, it does a merge.

\begin{figure}
\centerline{\includegraphics[scale=0.65]{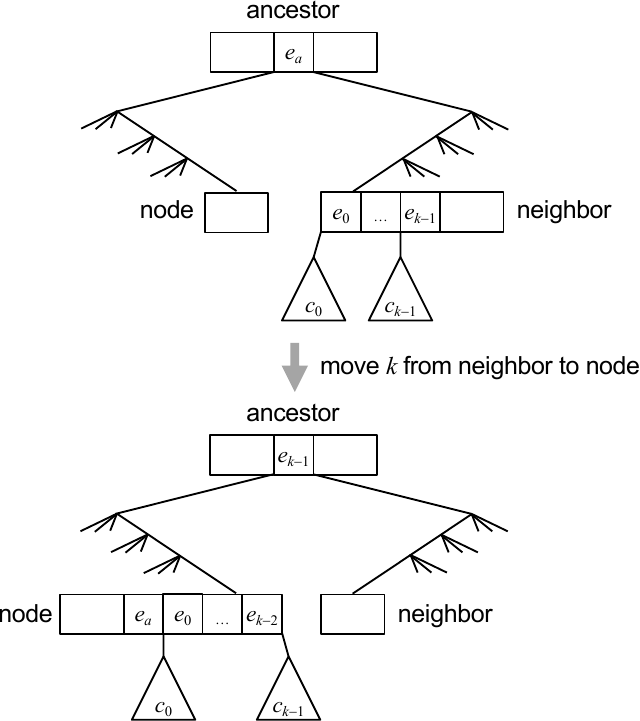}}
\caption{\label{fig:move_batch_general}Move batch.}
\end{figure}

Figure~\ref{fig:move_batch_general} illustrates the \emph{move}
operation, representing each pair $\tv{t_x}{v_x}$ of timestamp and
value as an entry~$e_x$.
In this figure, $k$ corresponds to \pseudocode{nodeDeficit}, i.e., the
number of entries and children to move to \pseudocode{node} to repair
its underflow by bringing its arity back to \pseudocode{MIN\_ARITY}.
In contrast to the textbook move
operation~\cite{cormen_leiserson_rivest_1990}, $k$ may exceed~1 and
\pseudocode{neighbor} may not be a sibling of \pseudocode{node}.
The only entry of the ordered window that is between \pseudocode{node}
and \pseudocode{neighbor} is $e_a$ in their least common ancestor.
So the move rotates $e_a$ into \pseudocode{node}, along with
\mbox{$e_0,\ldots,e_{k-2}$} and all associated children, and rotates
$e_{k-1}$ to the ancestor.
In the end, \pseudocode{node} has arity \pseudocode{MIN\_ARITY} and
\pseudocode{neighbor} has arity $\ge\pseudocode{MIN\_ARITY}$, because
it started with sufficient surplus.
Of course, \pseudocode{neighbor} still has arity
$\le\pseudocode{MAX\_ARITY}$, because it started out that way and did
not grow any bigger.
Figure~\ref{fig:move_batch_example}~\ifVldbElse{in the extended
version~\cite{tangwongsan_hirzel_schneider_2023}}{}shows pseudocode and a
concrete example for \emph{move}.

\begin{figure}
\centerline{\includegraphics[scale=0.65]{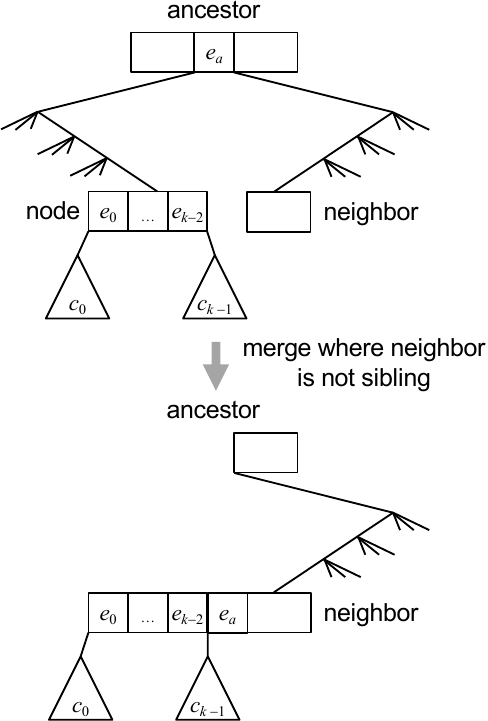}}
\caption{\label{fig:merge_notsibling_general}Merge with neighbor (non-sibling).}
\end{figure}

Figure~\ref{fig:merge_notsibling_general} illustrates \emph{merge},
which adds what is left of \pseudocode{node} to
\pseudocode{neighbor} and then eliminates \pseudocode{node}.
Unlike in the textbook B-tree setup, \pseudocode{node} and
\pseudocode{neighbor} may not be direct siblings.
Since any other vertices on the path from \pseudocode{node} to
\pseudocode{ancestor} are entirely $<t$, those vertices will also
be eliminated.
On the other hand, $e_a$ has a timestamp $>t$, so it remains in the
tree, and we rotate it into \pseudocode{neighbor}.
Let \pseudocode{oldNodeArity} and \pseudocode{oldNeighborArity} refer
to the arity of the node and its neighbor before the merge.
Then, after the merge, we have
\[\begin{array}{l@{\,}l@{\,}l}
&\hspace{-5mm} \pseudocode{neighbor.arity}\\
=~& \pseudocode{oldNodeArity}
  & +\,\pseudocode{oldNeighborArity}\\
=~ & \pseudocode{MIN\_ARITY} - \pseudocode{nodeDeficit}
  & +\,\pseudocode{MIN\_ARITY} + \pseudocode{neighborSurplus}\\
=~& \multicolumn{2}{@{\,}l}{2\cdot\pseudocode{MIN\_ARITY}
    + (\pseudocode{neighborSurplus} - \pseudocode{nodeDeficit})}\\
\end{array}\]
This means that there is no overflow, because merge only happens when
\mbox{$\pseudocode{nodeDeficit}>\pseudocode{neighborSurplus}$}, and
there is no underflow, because
\mbox{$\pseudocode{nodeDeficit}\le\pseudocode{MIN\_ARITY}$} and
\mbox{$\pseudocode{neighborSurplus}\ge0$}.
See code and example in
Figure~\ref{fig:merge_notsibling_example}\ifVldbElse{ in the extended
version~\cite{tangwongsan_hirzel_schneider_2023}}{}.

\begin{figure}
\begin{minipage}{0.48\columnwidth}
\centerline{\includegraphics[scale=0.65]{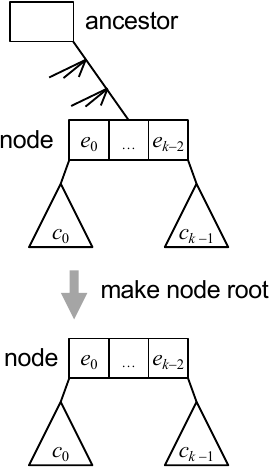}}
\caption{\label{fig:make_node_root_general}Make node root.}
\end{minipage}
\begin{minipage}{0.48\columnwidth}
\centerline{\includegraphics[scale=0.65]{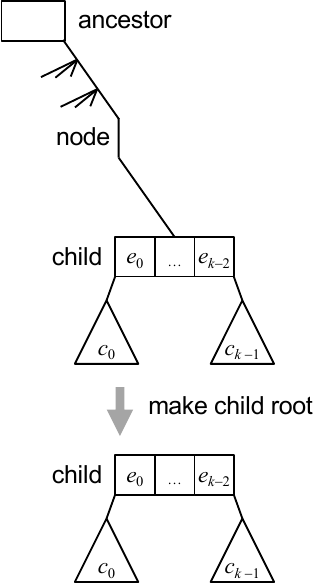}}
\caption{\label{fig:make_child_root_general}Make child root.}
\end{minipage}
\end{figure}

For \pseudocode{bulkEvict(}$t$\pseudocode{)} to be fully general, it
must handle the case where $t$ is all the way on the right spine.
This implies that the root itself is to the left of $t$ and must be
eliminated.
Eliminating the root shrinks the tree from the top, thus preserving
the height invariant, and requires giving the tree a new root lower down.
There are two sub-cases for shrinking the tree given a node on the
right spine.
If, after the local eviction, the node still has arity $>1$, the
algorithm makes it the root (Figure~\ref{fig:make_node_root_general});
otherwise, the node has arity $=1$ and the algorithm makes its single
child the root (Figure~\ref{fig:make_child_root_general}).
Figure~\ref{fig:make_child_root_example}~\ifVldbElse{in the extended
version~\cite{tangwongsan_hirzel_schneider_2023}}{}shows pseudocode and an
example.

\medskip
\noindent\textbf{Bulk eviction Step 3: Pass down.}
The last step of the algorithm repairs left aggregates and
spine flags on the left spine.
In case the eviction touched the right spine, it also repairs right
aggregates and spine flags on the right spine.
Recall that the left aggregate and right aggregate of a node are
computed using the aggregate result from its parent.
Hence, the pass down loops over tree levels and performs a local
recompute to propagate these changes.

\begin{theorem}\label{theorem_insert_time}
  The algorithm for \pseudocode{bulkEvict(}$t$\pseudocode{)} takes 
  $O(\log m)$ amortized time and $O(\log n)$ worst-case time.
\end{theorem}

\begin{proof}
  Consider the steps of the algorithm separately.
  Step~1, the finger-based search, takes time $O(\log m)$ worst-case,
  since it takes a single traversal up from a finger to the lowest ancestor
  containing $t$ followed by a single traversal down at most to a leaf.
  Step~2, the pass up, comprises an eviction loop followed by a repair loop.
  The eviction loop takes time $O(\log m)$ worst-case, since it
  traverses the \pseudocode{boundary} list returned by the search.
  The repair loop might continue to repair overflow past the top of
  the boundary.
  In the worst case, it might reach the root, bringing
  the total time complexity of the evict loop plus repair loop to
  $O(\log n)$ worst-case.
  However, since the repair loop starts above the boundary, at its
  start, it can at most have to deal with an underflow of a single entry.
  Therefore, it meets the conditions of Lemma~9 from the FiBA
  paper~\cite{tangwongsan_hirzel_schneider_2019}, which uses virtual
  coins to show that the amortized cost for the repair loop is~$O(1)$.
  This brings the total amortized time of the pass up to
  \mbox{$O(\log m+1)=O(\log m)$}.
  Finally, Step~3, the pass down, traverses the same number of levels
  as the pass up.
\end{proof}

%% \begin{figure*}
%% \begin{minipage}{0.68\textwidth}
%% \centerline{\includegraphics[scale=1.2]{move_batch_example.pdf}}
%% \caption{Move batch (example).}
%% \end{minipage}
%% \begin{minipage}{0.29\textwidth}
%% \centerline{\includegraphics[scale=0.6]{move_batch_general.pdf}}
%% \caption{Move batch (general).}
%% \end{minipage}
%% \end{figure*}

%% \begin{figure*}
%% \begin{minipage}{0.68\textwidth}
%% \centerline{\includegraphics[scale=1.2]{merge_notsibling_example.pdf}}
%% \caption{Merge where neighbor is not sibling (example).}
%% \end{minipage}
%% \begin{minipage}{0.29\textwidth}
%% \centerline{\includegraphics[scale=0.6]{merge_notsibling_general.pdf}}
%% \caption{Merge where neighbor is not sibling (general).}
%% \end{minipage}
%% \end{figure*}

\section{Bulk Insertion}\label{sec:bulkinsertion}

As defined in Section~\ref{sec:problem},
\pseudocode{bulkInsert(}$B^\textrm{in}$\pseudocode{)} inserts one or more entries into the
window. The bulk of entries is modeled as an iterator of
\pseudocode{(timestamp, value)} pairs, which are assumed to be
timestamp-ordered. Our \pseudocode{bulkInsert} algorithm processes the bulk in
three steps:
\begin{enumerate}[label={\textbf{Step~\arabic*}}, leftmargin=*]
  \item A finger-based \emph{insertion sites search} that, without
    making any modifications, locates all the sites in the tree where
    new entries need to be inserted.
  \item A \emph{pass up: interleave\&split loop} that, starting at the leaves,
    interleaves the
    new entries into their respective nodes, splitting the node and promoting
    keys as necessary to satisfy the arity invariants. This happens from the
    leaves up until no level requires further processing.
  \item A~\emph{pass down} the right spine, and if needed also the
    left spine, that repairs any leftover aggregation invariant violations.
\end{enumerate}
The remainder of this section delves deeper into the details of these steps and
their cost analysis. Later, Section~\ref{sec:implementation} discusses their
implementation and optimization maneuvers.

\medskip
\noindent\textbf{Bulk insertion Step 1: Insertion sites search.}
To locate the insertion sites, the algorithm conducts the search in
timestamp order, beginning with the earliest timestamp in the bulk
using finger search.
Each subsequent search never has to go higher than the least common
ancestor between the previous node and its insertion site.
This step associates each \pseudocode{(timestamp, value)} pair from the
input with the corresponding node into which it will be inserted.

Like in a standard B-tree structure, each new timestamp (key) that is not yet
in the tree will be inserted at a leaf location. Such a key can cause
cascading changes to the tree structure and, in the context of FiBA, can
additionally trigger a chain of recomputation of aggregation values starting
from the insertion site. On the other hand, a timestamp that is already in the
tree is destined to the node where that timestamp is present, where
the aggregation monoid combines its value
with the existing value. This results in no structural changes, but
in the context of FiBA, this triggers a chain of recomputation of aggregation
values starting from that node. We see both cases as events that require
processing:
an \emph{insertion event} adds a real entry to the target node and
recomputes the aggregation value,
whereas a \emph{recomputation event} merely indicates the node where
recomputation must take place.

\medskip
\noindent$\rhd$~\textit{Treelets.}
Concretely, the implementation represents each event as a
\emph{treelet} tuple \pseudocode{(target, timestamp, value, childNode, kind)}.
This indicates that this particular
\pseudocode{(timestamp, value)} pair with a child \pseudocode{childNode}
(possibly \pseudocode{NULL}) is to be inserted into the \pseudocode{target}
node unless the \pseudocode{kind} is a recomputation event, in which case it
simply triggers a recomputation of aggregate values on the \pseudocode{target}
node. Treelets form the backbone of the \pseudocode{bulkInsert} logic, with
Step~1 (the insertion sites search) creating the initial timestamp-ordered sequence of treelets
targeting all the relevant insertion sites. 

\begin{figure}
  \includegraphics[width=2in]{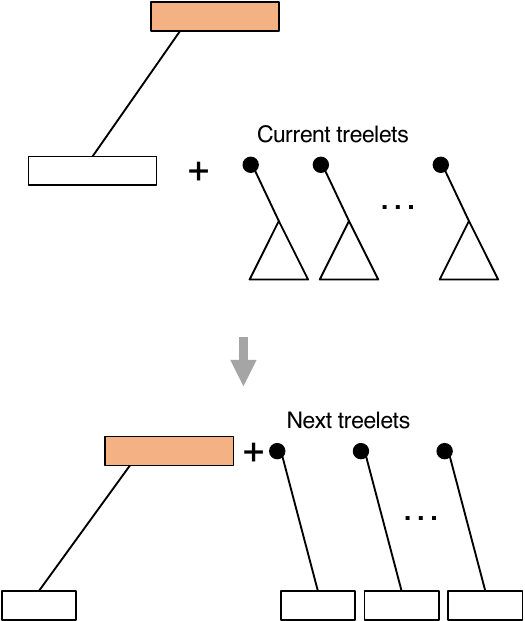}
\caption{Interleave and split for one level of a tree}
\label{fig:interleave-and-split}
\end{figure}
\medskip
\noindent\textbf{Bulk insertion Step 2: Pass up: interleave\&split loop.}
As the next step, the algorithm
proceeds level by level, working its way from the leaf level towards the root
until no more changes happen. At any point, the algorithm aims to maintain only
two levels of treelets---the current level and the next level. In this view, as
illustrated in Figure~\ref{fig:interleave-and-split}, each level takes as input a
sequence of treelets and produces a sequence of treelets for the next level.
Since the treelets in the input are timestamp-ordered, the entries destined for
the same node appear consecutively in the sequence and are easily identified.
Conceptually, each level is processed as follows:
\begin{quote}
\noindent{}For each target $t$ in the input sequence of treelets:
  \begin{enumerate}[label=(\roman*), leftmargin=2em]
  \item Gather all the treelets that target $t$ into \pseudocode{TL}.
  \item Interleave the contents of $t$ with \pseudocode{TL}.
    Since both of these are ordered, the interleave routine is the
    merge step of the well-known merge-sort algorithm.
    Interleaving takes time linear in the total length of its
    input sequences to produce an ordered output sequence without
    requiring a separate sort step.
  \item If $t$ has arity more than \pseudocode{MAX\string_ARITY}, apply
  \pseudocode{bulkSplit} to split it into smaller nodes.
  \end{enumerate}
\end{quote}

When multiple entries are added to the same node, a node can temporarily
overflow to arity $p > \pseudocode{MAX\string_ARITY}=2\mu$, often $p \gg 2\mu$.
The \pseudocode{bulkSplit} routine then splits it into invariant-respecting nodes,
consisting of one or more arity-$(\mu+1)$ nodes and one last node with arity between
$\mu$ and $2\mu$.  The following claim, which is intuitive and whose proof
appears in~\ifVldbElse{the extended
paper~\cite{tangwongsan_hirzel_schneider_2023}}{the appendix}, shows that it is possible
to split such a node into legitimate FiBA nodes in this way:
\begin{claim}
  \label{claim:well-splittable}
  Let $p > \pseudocode{MAX\string_ARITY}~=2\mu$ be an integral temporary arity.
  The number $p$ can be written as 
  \[
    p = b_0 + b_1 + \dots + b_{t-1} + b_t,
  \]
  where $b_0 = b_1 = \dots b_{t-1} = \mu+1$ and $\mu \leq b_t \leq 2\mu$.
\end{claim}
For example, if $p = 2\mu + 3$ with $\mu = 4$, we can write $p$ as
\mbox{$p = (\mu + 1) + (\mu + 2)$}.
That is, this split yields one arity-$(\mu+1)$ node, one entry to send up
to the next level, and one arity-$(\mu+2)$ node.
If $p = 7\mu
+ 2$ with $\mu = 2$, we can write $p$ as $p = \mu + 1 + \mu + 1 + \mu + 1 + \mu
+ 1 + 2\mu$. That is, this split yields four arity-$(\mu+1)$ nodes and one arity-$2\mu$
node, interspersed with $4$ entries to send up to the next level. 

\medskip
\noindent$\rhd$~\textit{Promotion to the next level.} Splitting an overflowed
node also generates treelets, representing entries promoted for insertion into
nodes in the next level. Importantly, by processing current-level treelets in
timestamp order, new treelets for the next level generated in this manner are
already sorted in timestamp order. This helps avoid the costly step of sorting
them or the need for a priority queue. Additionally, the parent of each
existing node is the target insertion site of the corresponding promoted entry.

The discussion so far left out recomputation events.
There are two ways a recomputation event is created:
(a)~inserting an entry with an existing timestamp and
(b)~incorporating entries into a node without causing it to overflow.
Case~(a) happens in
Step~1 (insertion sites search) but can target nodes anywhere in the tree, not
just the leaves. Case~(b) happens throughout Step~2 (making a pass up).
Because of how Step~1 is carried out and to sidestep the need to
store treelets for future levels and interleave in treelets for recomputation
events when their levels are reached, we start all the recomputation
events/treelets at the leaf level. These treelets will ride along with the
other treelets but will not have a real effect until their levels are reached. This
turns out to have the same asymptotic complexity as if we were to start them at
their true levels---but without the additional code complexity.

\medskip
\noindent\textbf{Bulk insertion Step 3: Pass down.}
Like in the
\pseudocode{bulkEvict} algorithm, the final step repairs right aggregates on
the right spine and potentially left aggregates on the left spine if it also
touches the left spine. For both spines, the aggregate of a node is computed
using the value from its parent, so this computation is a pass on the spine
towards the finger (i.e., rightmost and leftmost leaf).

\medskip
\noindent\textbf{Bulk insertion: Time complexity analysis.}
The time complexity of
\pseudocode{bulkInsert} can be broken down into (i) the search cost (Step~1),
(ii) insertion and tree restructuring (Step~2), and (iii) aggregation repairs
(during Steps~2 and 3).  To analyze this, we begin by proving a lemma that
quantifies the footprint---the worst-case number of nodes that can be
affected---when there are $m$ insertion sites.

For a \pseudocode{bulkInsert} call, the \emph{top} node, denoted by $\tau$, is
the least-common ancestor of all insertion sites and the rightmost finger.
By definition, this is the node closest to the
leaf level where paths from all these sites towards the root converge.
\begin{lemma}
  \label{lem:insertion-footprint}
  In a FiBA structure with \pseudocode{MAX\string_ARITY}$\,=2\mu$, if there are
  $m$ insertion sites, the paths from all the insertion sites, as well as the
  node at the right finger, to the top $\tau$ contain at most $O(m ( 1 + \log_{2\mu}
  (\tfrac{N_\tau}{m})))$ unique nodes, where $N_\tau$ is the total number of nodes in the
  subtree rooted at $\tau$.
\end{lemma}

\begin{proof}
  Consider the subtree rooted at the top node $\tau$. For level $\ell = 0, 1,
  \dots$ away from the top, the total number of nodes at that level $n_\ell$
  satisfies 
  \begin{equation}
    \mu^\ell \leq n_\ell \leq (2\mu)^\ell,
    \label{eq:num_node_bound}
  \end{equation}
  which holds because the fan-out degree for non-root\footnote{If $\tau$ is the
  root, the bound is slightly different since the root can have as few as two
  children, but the statement of the lemma remains the same.} nodes is
  between $\mu$ and $2\mu$ (inclusive). Now we will assume all the insertion
  sites are at the leaf level. This can be arranged by projecting every
  insertion site onto a leaf within its own subtree, and doing so can only
  increase the number of nodes contributing to the bound. 
  
  By \eqref{eq:num_node_bound}, the leaves must be at level $L \leq \log_{\mu}
  N_\tau$ and the smallest level $\ell$ that has no more than $m$ nodes is $\ell
  \geq \log_{2\mu} m$. This means the paths from the leaf insertion sites can
  travel without necessarily converging together for $L - \ell$ levels.  During
  this stretch, the number of unique nodes is at most $m(L - \ell) =
  O(m\log_{2\mu}(N_\tau/m))$. From level $\ell$ to the top node, the paths must
  converge as constrained by the shape of the tree. In a B-tree
  with $m$ leaves, the number of nodes at each level decreases geometrically
  towards the top. Hence, there are at most $O(m)$ unique nodes from levels
  $\ell$ and above, for a grand total of $O(m(1 +
  \log_{2\mu}(\tfrac{N_\tau}{m})))$ unique nodes.
\end{proof}

Next, we address the tree restructuring cost:

\begin{lemma}
  \label{lem:fiba-bulk-restructuring}
  Let $\mu \geq 2$. The tree restructuring cost of inserting $m$ entries in a
  bulk is amortized $O(m)$ and worst-case $O(m\log (\frac{m+n}{m}))$, where $n$
  is the number of entries prior to the bulk insertion operation.
\end{lemma}

The worst-case bound can be easily seen: the $m$ insertions can only change the
nodes from $m$ leaves to the root, touching at most $O(m\log (\frac{m+n}{m}))$
nodes (Lemma~\ref{lem:insertion-footprint}). For the amortized bound, the proof
is analogous to Lemma~9 in the FiBA
paper~\cite{tangwongsan_hirzel_schneider_2019}, arguing that charging $2$ coins
per new entry is sufficient in maintaining the tree. More details appear
in~\ifVldbElse{the extended paper~\cite{tangwongsan_hirzel_schneider_2023}}{the appendix}.

\begin{theorem}
  The algorithm for \pseudocode{bulkInsert} runs in amortized $O(\log d + m( 1 + \log
  (\tfrac{d}{m})))$ time and $O(\log d + m \log (\frac{m+n}{m}))$ worst-case time, where $m$ is
  the number of entries in the bulk and $d$ is the out-of-order distance of the
  earliest entry in the bulk.
\end{theorem}

\begin{proof}
  The running time of \pseudocode{bulkInsert} is made up of (i) the search cost
  (Step~1), (ii) insertion and tree restructuring (Step~2), and (iii) aggregation repairs
  (during Steps~2 and 3).
  
  The first search for the insertion site takes $O(\log d)$, thanks
  to finger searching from the right finger. Each subsequent search only
  traverses the path from the previous entry to their least common ancestor and
  down to the next entry. The whole search cost is therefore covered by
  Lemma~\ref{lem:insertion-footprint}. After that, the actual insertion takes
  $O(1)$ time per entry since interleaving takes time that is
  linear in its input. The cost to further restructure the tree is as described
  in Lemma~\ref{lem:fiba-bulk-restructuring}. Finally, it is easy to see that
  the cost of aggregation recomputation/repairs is subsumed by the first two
  costs because the aggregation of a node has to be recomputed only if it was
  part of the restructuring or sits on the search path (spine or on the way to
  the top node). Adding up the costs yields the stated bounds.
\end{proof}

This means asymptotically \pseudocode{bulkInsert} is never more
expensive than individually inserting entries. On the contrary, bulk insertion
results in cost savings as insertion-site search and restructuring work can be
shared.

\section{Implementation}\label{sec:implementation}

We implemented our algorithm in C++ because of its strong and
predictable raw performance in terms of both time and space.
Using C++ avoids latency spikes from runtime services, such as garbage
collection or just-in-time compilation, common in managed
languages such as Java or Python.
Such extraneous latency spikes would obscure the latency effects of
our algorithm.
The results section contains apples-to-apples comparisons with
other sliding-window aggregation algorithms from prior work that was
also implemented in C++.
We reuse code between our new algorithm and those earlier algorithms.
In particular, we use C++ templates to specialize each algorithm for
each given aggregation monoid, and share the same implementation of
the aggregation monoids across all algorithms.
The C++ compiler then inlines both the monoid's data structure and its
operator code into the sliding-window aggregation data structure and
algorithm code as appropriate.

\medskip
\noindent{}\textbf{Deferred free list.} 
Our implementation has to avoid reclaiming memory eagerly.
If bulk eviction reclaimed memory eagerly, the promised algorithmic
complexity would be spoiled: Given that the arity of the tree is controlled by a constant
hyperparameter \pseudocode{MIN\_ARITY}, eagerly evicting a bulk of $m$
entries would require reclaiming the memory of $O(m)$ nodes.
Those $O(m)$ calls to \pseudocode{delete} would be worse than the
amortized complexity of $O(\log m)$ for bulk evict.
Therefore, we avoid eager memory reclamation as follows:
Recall that the eviction loop iterates over $O(\log m)$ nodes on the
boundary and, for each node, performs local evictions, which will proceed to
evict the children of that node.
Instead of recursively deleting all the descendants eagerly, the local evict
places their children on a deferred free-list.
Since at most $O(\log m)$ nodes can be removed, the cost of adding only the
children to the free list during bulk eviction is worst-case $O(\log m)$.
Later, when an insertion would require a new allocation, it first
checks the free-list.
If that is non-empty, it pops one node, pushes its children, and
reuses its memory for the new node.
Thus, each insert only spends worst-case $O(1)$ time on memory reuse.

\smallskip
\noindent{}\textbf{Memory mangement during \mbox{\normalfont\pseudocode{bulkInsert}.}} 
Conceptually, we allow a node to grow to an arbitrary size before splitting it
into invariant-respecting smaller nodes. For performance, the implementation
does this differently. The main goal is to minimize memory allocation and
deallocation for intermediate storage. 
To combine keys from an existing node with keys to be inserted into that node, it
employs an ordered interleaving routine from merge sort. Here the interleaving is lazy:
instead of generating the combined sequence of keys upfront, our
implementation offers an iterator for the interleaved sequence that computes the
next element on the fly, reading directly from two sources---the existing node
and the sequence of treelets for the current level.  We also have an
optimization where if the node is not going to overflow after incorporating the
new keys (``small insertion''), then simple insertion is used as there would be
no memory allocation involved.  

Additional optimization includes (i) using alternating buffers for treelet
processing and (ii) consolidating treelets. For treelet processing, the
dataflow pattern is reading from the current level and writing to the next
level. Each sequence is progressively smaller as the algorithm works its way up
the tree. Hence, we allocate two \pseudocode{vector}s with enough capacity at
the start and alternate between them as the algorithm proceeds. Furthermore,
treelets that will be inserted into the same node are consolidated together.
This reduces the struct size because the target node does not need to be
repeated for each of these treelets.

\smallskip
\noindent{}\textbf{Miscellanea.}
As described, our algorithm already combines entries with the same timestamp at
insert, thus reducing memory.  Users can choose to coarsen the granularity of
timestamps, thus causing more cases of equal timestamps, recovering basic
batching. However, it would require additional work to take full advantage of
batching, such as for energy efficiency~\cite{michalke_et_al_2021}.
Our implementation does not directly use SIMD instructions, but the C++
optimizing compiler sometimes uses them automatically.
We did not implement partitioning but it is straightforward: when the aggregate
is partitioned by key, keep disjoint state, i.e., a separate tree for each key;
that would enable fission~\cite{hirzel_schneider_gedik_2017} for parallelization, 
either user-directed or automatically.
Previous work describes an algorithm for range
queries~\cite{tangwongsan_hirzel_schneider_2019}, and that algorithm also works
in the presence of bulk insertion and eviction. Future work could pursue a new
algorithm for multi-range queries.

\begin{figure*}
\includegraphics[width=.32\textwidth]{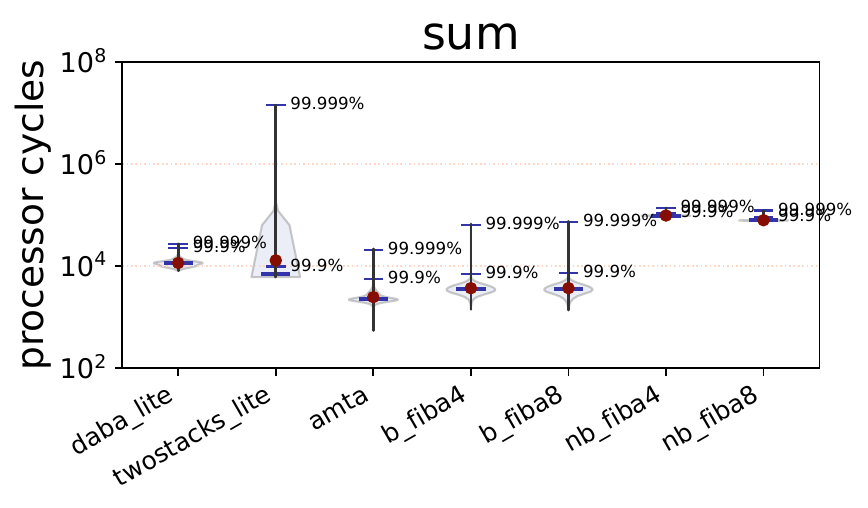}%
\hfill%
\includegraphics[width=.32\textwidth]{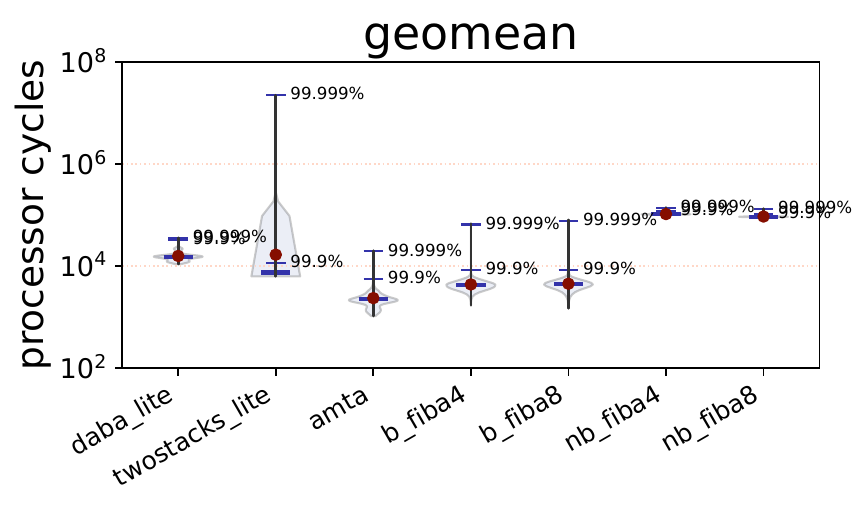}%
\hfill%
\includegraphics[width=.32\textwidth]{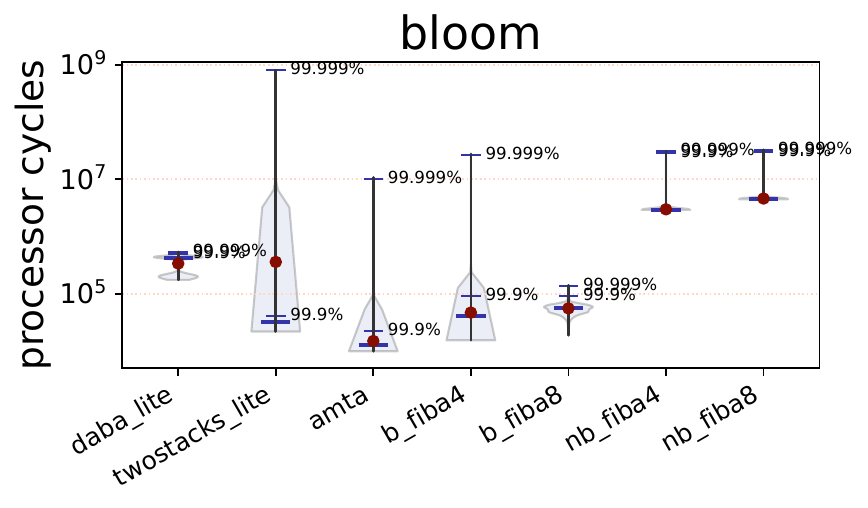}
\vspace*{-5mm}
\caption{\label{fig:latency_bulk_evict_opevict_w4194304_d0_b1024}Latency, bulk evict only, window size $n=\textrm{4,194,304}$, bulk size $m=\textrm{1,024}$, in-order data $d=\textrm{0}$.}
\vspace*{-2mm}
\end{figure*}

\begin{figure*}
\includegraphics[width=.32\textwidth]{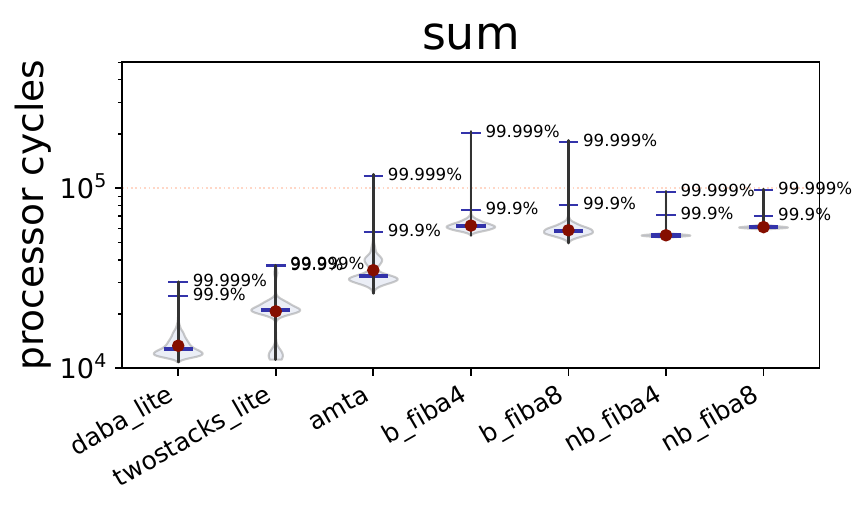}%
\hfill%
\includegraphics[width=.32\textwidth]{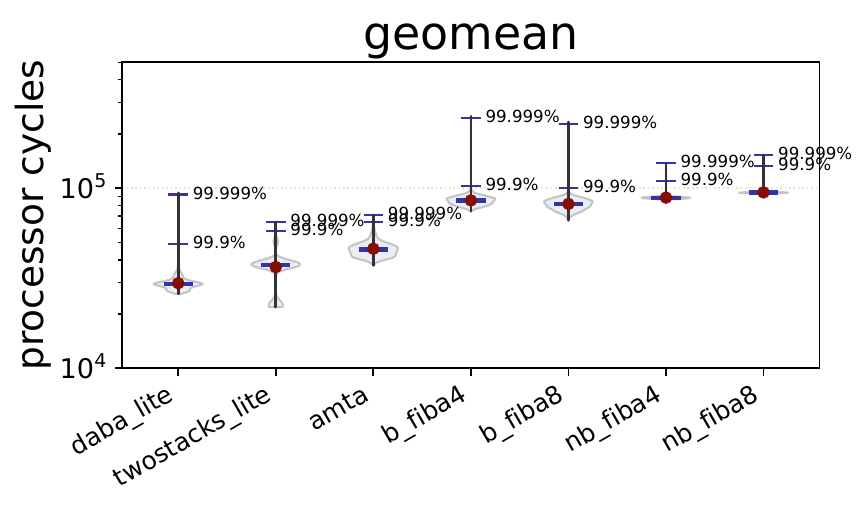}%
\hfill%
\includegraphics[width=.32\textwidth]{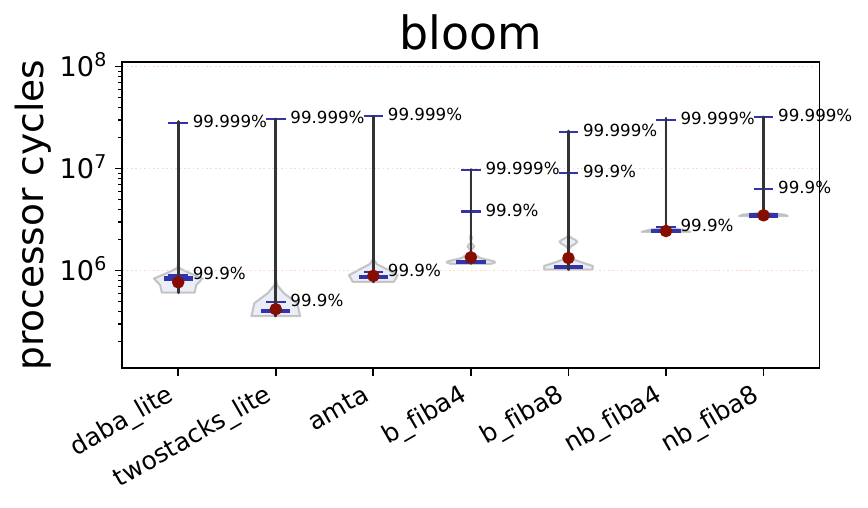}
\vspace*{-5mm}
\caption{\label{fig:latency_bulk_evict_insert_opinsert_w4194304_d0_b1024}Latency, bulk insert only, window size $n=\textrm{4,194,304}$, bulk size $m=\textrm{1,024}$, in-order data $d=\textrm{0}$.}
\end{figure*}

\section{Results}\label{sec:results}

%% \todo{R4(D7). You could improve your evaluation by discussing findings
%%   in more detail. Currently, some additional clarification by how much
%%   you outperform the baselines and why is needed.}

%% \todo{R4(D8). Overall, your figures and charts could be easier to read
%%   on gray-scale. In particular, the individual lines in your plots are
%%   hard to distinguish.}

This section explores how the theoretical algorithmic complexity
from the previous sections play out in practice.
It explores how performance correlates with
the number $n$ of entries in the window,
the number $m$ of entries in the bulk insert or bulk evict, and
the number $d$ of entries between an insertion and the youngest
end of the window.
The experiments use multiple monoidal aggregation operators
to cover a spectrum of computational cost: \pseudocode{sum}~(fast),
\pseudocode{geomean}~(medium), and
\pseudocode{bloom}~\cite{bloom_1970}~(slow).

This section refers to different
sliding-window aggregation algorithms as follows:
The original non-bulk FiBA
algorithm~\cite{tangwongsan_hirzel_schneider_2019} is
\pseudocode{nb\_fiba4} and \pseudocode{nb\_fiba8}, with
\pseudocode{MIN\_ARITY} of 4 or~8.
Similarly, the new bulk FiBA algorithm introduced in this paper is
\pseudocode{b\_fiba4} and \pseudocode{b\_fiba8}.
Both of these algorithms can handle out-of-order data.
As baselines, several figures include three algorithms
that only work for in-order data, i.e., \mbox{when $d=0$}.
The amortized monoid tree aggregator, \pseudocode{amta}, supports bulk
evict but not bulk insert~\cite{villalba_berral_carrera_2019}.
The \pseudocode{twostacks\_lite} algorithm performs single insert or
evict operations in amortized $O(1)$ and worst-case $O(n)$
time~\cite{tangwongsan_hirzel_schneider_2021}.
The \pseudocode{daba\_lite} algorithm performs single insert or
evict operations in worst-case $O(1)$
time~\cite{tangwongsan_hirzel_schneider_2021}.
Since \pseudocode{amta}, \pseudocode{twostacks\_lite}, and
\pseudocode{daba\_lite} require in-order data, they are absent from
figures with results for out-of-order scenarios.

We ran all experiments on a machine with dual Intel Xeon Silver 4310
CPUs at 2.1~GHz running Ubuntu 20.04.5 with a 5.4.0 kernel. We compiled
all experiments with  \verb@g++@ 9.4.0 with optimization level \verb@-O3@.
To reduce timing noise and variance in memory
allocation latencies, we use \verb@mimalloc@~\cite{leijen_zorn_demoura_2019}
instead of the stock
glibc allocator, and we pin all runs to core~0 and the corresponding
NUMA group.

\subsection{Latency}\label{sec:results_latency}

In streaming applications, late results are often all but useless: with
increasing latency, the value of a streaming computation reduces
sharply---for example, dangers become too late to avert and opportunities are
missed.
% For example, if the application reacts too late to some danger, that
% danger becomes impossible to avert.
% Similarly, if the application reacts too late to some opportunity,
% that opportunity passes.
Therefore, our algorithm is designed to support both the finest
granularity of streaming~(i.e., when $m=1$) as well as bursty
data~(i.e., when~$m\gg 1$) with low latencies.
Even in the latter case, our algorithm still retains the ability of
tuple-at-a-time streaming, unlike systems with a micro-batch model.
The methodology for the latency experiments is to measure how long
each individual insert or evict takes, then visualize the distribution
of insertion or eviction times for an entire run as a violin plot.
The plots indicate the arithmetic mean as a red dot, the median as a
thick blue line, and the 99.9\textsuperscript{th} and
99.999\textsuperscript{th} percentiles as thin blue lines.
At 2.1~GHz, $10^4$ processor cycles correspond to 4.8~microseconds.

\textbf{Figure~\ref{fig:latency_bulk_evict_opevict_w4194304_d0_b1024}}
shows the latencies for bulk evict with in-order data.
This experiment loops over
evicting the oldest $m=1,024$ entries in a single bulk,
inserting 1,024 new entries one by one,
and calling \pseudocode{query},
measuring only the time that the bulk evict takes.
In theory, we expect bulk evict to take time $O(\log m)$ for
\pseudocode{b\_fiba4} and \pseudocode{b\_fiba8}, and $O(\log n)$ for
\pseudocode{amta}.
The remaining algorithms, lacking a native bulk evict, loop over
single evictions, taking $O(m)$ time.
In practice, \pseudocode{b\_fiba4}, \pseudocode{b\_fiba8}, and
\pseudocode{amta} have the best latencies for this experiment,
confirming the theory.

\textbf{Figure~\ref{fig:latency_bulk_evict_insert_opinsert_w4194304_d0_b1024}}
shows the latencies for bulk insert with in-order data.
This experiment loops over
evicting the oldest $m=1,024$ entries in a single bulk,
inserting $m=1,024$ new entries in a single bulk,
and calling \pseudocode{query},
measuring only the time that the bulk insert takes.
In theory, since $d=0$ in this in-order scenario, the complexity of
bulk insert boils down to $O(m)$ for all considered algorithms.
In practice, \pseudocode{daba-lite} and \pseudocode{twostacks-lite}
yield the best latencies for this scenario since they incur no extra
overhead to be ready for an out-of-order case that does not occur here.

\begin{figure*}
\includegraphics[width=.32\textwidth]{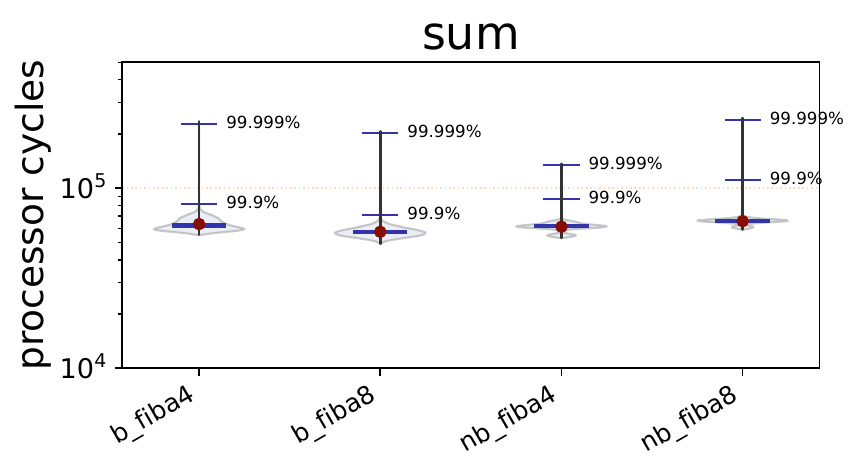}%
\hfill%
\includegraphics[width=.32\textwidth]{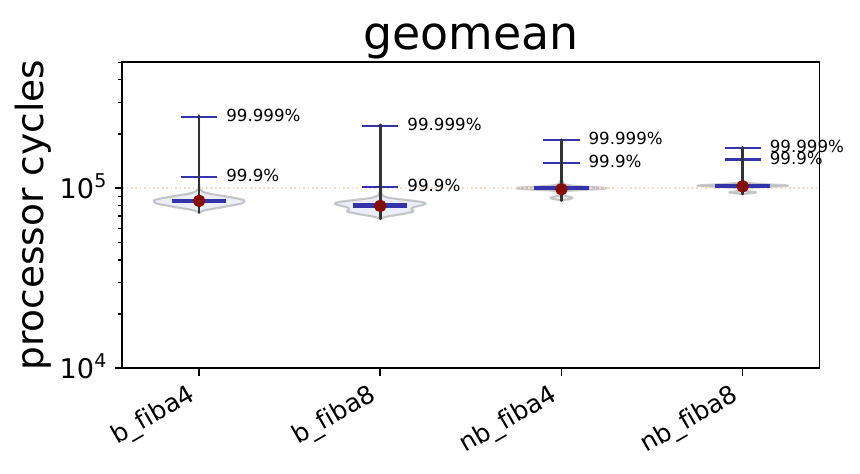}%
\hfill%
\includegraphics[width=.32\textwidth]{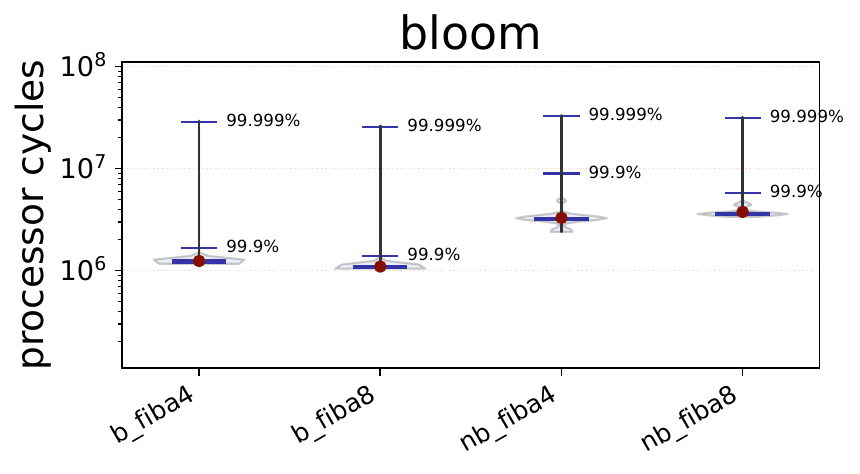}
\vspace*{-5mm}
\caption{\label{fig:latency_bulk_evict_insert_opinsert_w4194304_d1024_b1024}Latency, bulk insert only, window size $n=\textrm{4,194,304}$, bulk size $m=\textrm{1,024}$, out-of-order data $d=\textrm{1,024}$.}
\end{figure*}

\textbf{Figure~\ref{fig:latency_bulk_evict_insert_opinsert_w4194304_d1024_b1024}}
shows the latencies for bulk insert with out-of-order data.
This experiment differs from the previous one in that each bulk insert
happens at a distance of $d=1,024$ from the youngest end of the window.
Since \pseudocode{amta}, \pseudocode{twostacks\_lite}, and
\pseudocode{daba} only work for in-order data, they cannot participate
in this experiment.
In theory, we expect bulk insert to take $O(m\log\frac{d}{m})$ for
\pseudocode{b\_fiba} and $O(m\log d)$ for \pseudocode{nb\_fiba}, which
is worse.
In practice, \pseudocode{b\_fiba} has lower latency than
\pseudocode{nb\_fiba}, confirming the theory.

%% This is because our new algorithm speeds things up by not repeating
%% the $O(\log d)$ search $m$ times.

\begin{figure}
\centerline{\includegraphics[width=.32\textwidth]{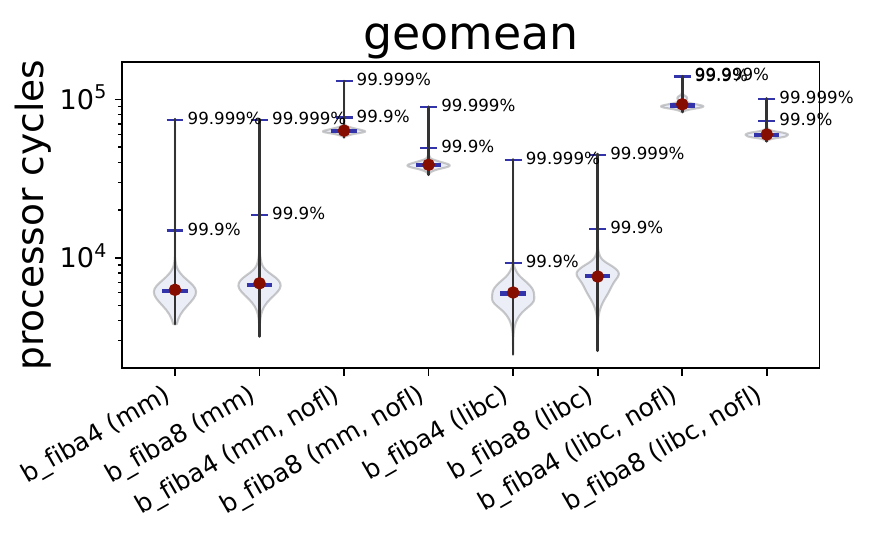}}
\vspace*{-4mm}
\caption{\label{fig:ablation}Memory management ablation study (latency, bulk evict only, $n=\textrm{4,194,304}$, $m=\textrm{4,096}$, $d=0$).}
\vspace*{-4mm}
\end{figure}

\textbf{Figure~\ref{fig:ablation}} shows an ablation
  experiment for memory-manage\-ment related implementation details.
  It compares results with mimalloc~(\pseudocode{mm}) vs.\ the default memory
  allocator~(\pseudocode{libc}), and with or without~(\pseudocode{nofl}) the deferred free list
  from Section~\ref{sec:implementation}.
  Consistent with the theory, the deferred
  free list is indispensible: \pseudocode{nofl} performs much worse. On the other hand,
  mimalloc made little difference; we use it to control for events that are so
  rare that they did not manifest in this experiment.

\subsection{Throughput}\label{sec:results_throughput}

%% Besides latency, another important metric for evaluating streaming
%% algorithms is
Throughput is the number of items in a long but
finite stream divided by the time it takes to process that stream.
The throughput experiments thus do not time each insert or evict
operation individually.
While the time for each individual operation may differ, we already saw those
distributions in the latency experiments, and here we focus on the gross
results instead.
%% As mentioned earlier, our algorithm is designed to support streaming
%% at the finest granularity.
%% Therefore,
The experiments include a memory fence before every insert to prevent the
compiler from optimizing (e.g., using SIMD) across multiple stream data items,
as that would be unrealistic in fine-grained streaming.
All throughput charts show error bars based on repeating each run
five times.

\begin{figure*}
\centerline{\includegraphics[width=\textwidth]{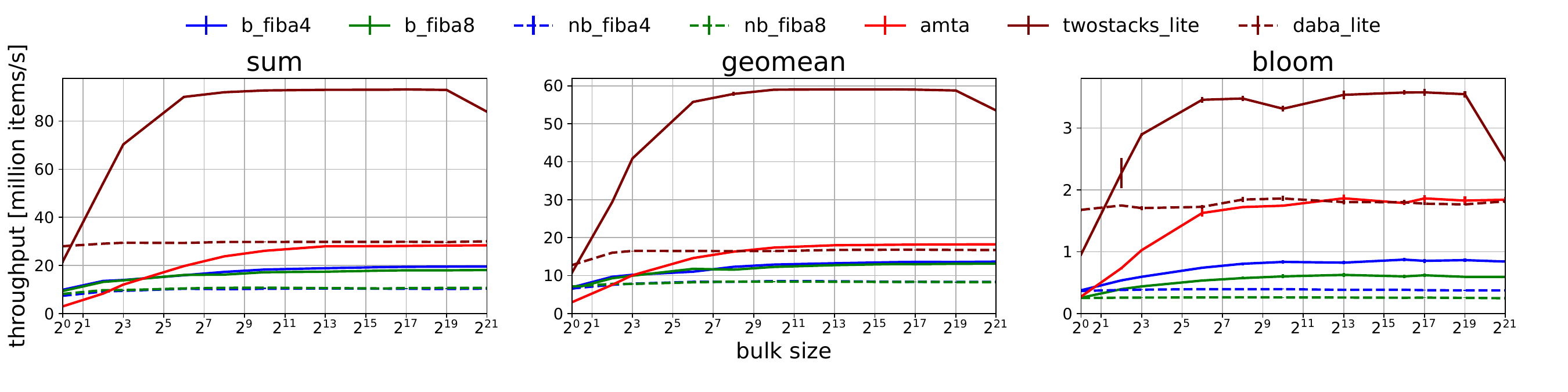}}
\vspace*{-4mm}
\caption{\label{fig:bulk_evict_window4194304}Throughput, bulk evict only, window size $n=\textrm{4,194,304}$, varying bulk size $m$, in-order data $d=\textrm{0}$.}
\end{figure*}

\textbf{Figure~\ref{fig:bulk_evict_window4194304}}
shows the throughput for running with bulk evict for in-order data as
a function of the bulk size~$m$.
This experiment loops over
a single call to \pseudocode{bulkEvict} for the oldest $m$ entries,
$m$ calls to single \pseudocode{insert},
and a call to \pseudocode{query}.
The throughput is computed from the time for the entire run, which
includes all these operations.
In theory, we expect the throughput of \pseudocode{b\_fiba} and
\pseudocode{amta} to improve with larger bulk sizes as they
have native bulk eviction. In practice, while that is true, even for algorithms that do not
natively support bulk evict, throughput also improves with larger~$m$.
This may be because their internal loop for emulating bulk evict
benefits from compiler optimization.
For in-order data, \pseudocode{twostacks\_lite} yields the
best throughput (but not the best latency, see
Section~\ref{sec:results_latency}).

\begin{figure*}
\centerline{\includegraphics[width=\textwidth]{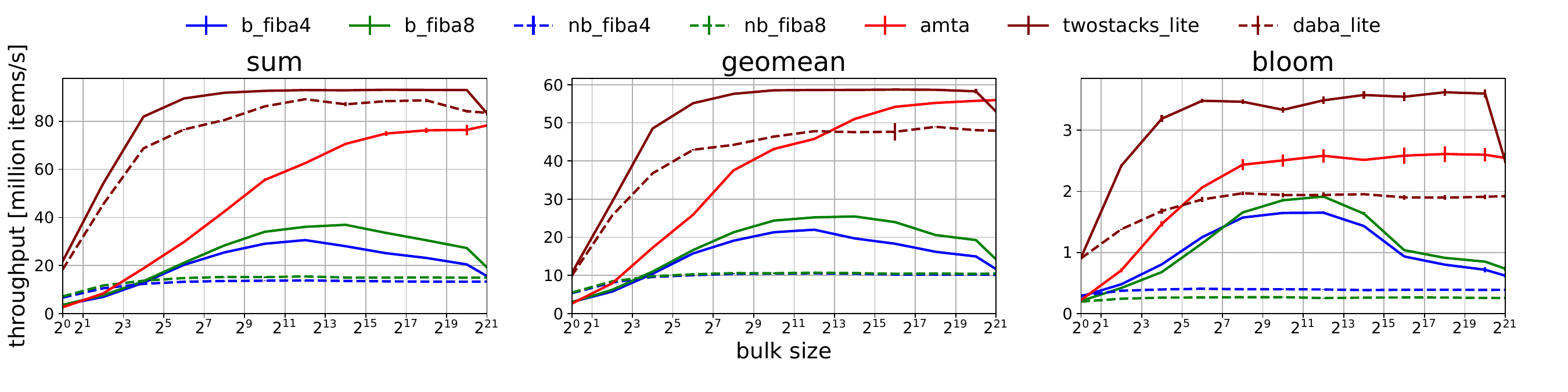}}
\vspace*{-4mm}
\caption{\label{fig:bulk_evict_insert_window4194304}Throughput, bulk evict+insert, window size $n=\textrm{4,194,304}$, varying bulk size $m$, in-order data $d=\textrm{0}$.}
\end{figure*}

\textbf{Figure~\ref{fig:bulk_evict_insert_window4194304}}
shows the throughput for running with both bulk evict and bulk insert
for in-order data as a function of the bulk size~$m$.
In theory, we expect that since the data is in-order, bulk insert
brings no additional advantage over looping over single inserts.
In practice, all algorithms improve in throughput as $m$ increases
from $2^0$ to around $2^{12}$.
This may be because fewer top-level insertions means fewer memory
fences, even for algorithms that emulate bulk insert with loops.
Furthermore, throughput drops when $m$ gets very large, because the
implementation needs to allocate more temporary space to hold data
items before they are inserted in bulk.

\begin{figure*}
\centerline{\includegraphics[width=\textwidth]{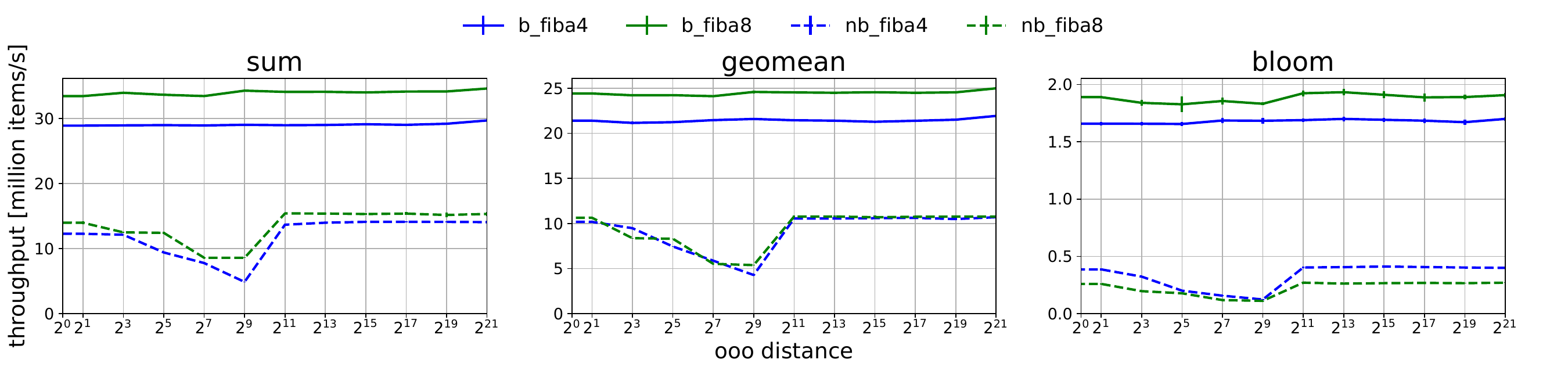}}
\vspace*{-4mm}
\caption{\label{fig:bulk_evict_insert_bulksize1024}Throughput, bulk evict+insert, window size $n=\textrm{4,194,304}$, bulk size $m=\textrm{1,024}$, varying ooo distance $d$.}
\end{figure*}

\textbf{Figure~\ref{fig:bulk_evict_insert_bulksize1024}}
shows the throughput as a function of the out-of-order degree~$d$ when
running with both bulk evict and bulk insert.
The \pseudocode{amta}, \pseudocode{twostacks\_lite}, and
\pseudocode{daba} algorithms do not work for out-of-order data and
therefore cannot participate in this experiment.
In theory, we expect that thanks to only doing the search once per
bulk insert, higher $d$ should not slow things down.
In practice, we find that that is true and \pseudocode{b\_fiba}
outperforms \pseudocode{nb\_fiba}.

\begin{figure*}
\centerline{\includegraphics[width=\textwidth]{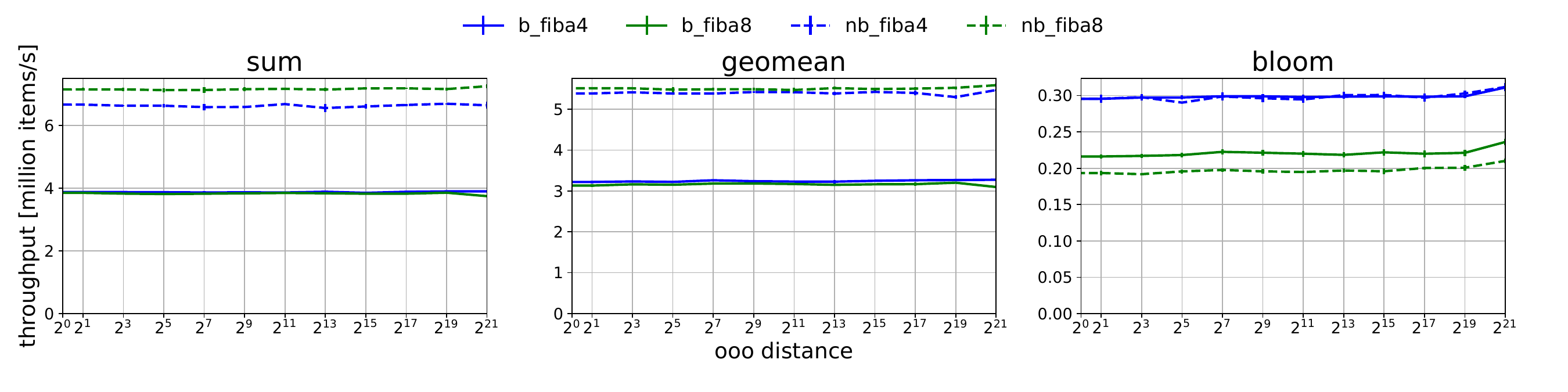}}
\vspace*{-4mm}
\caption{\label{fig:bulk_evict_insert_bulksize1}Throughput, bulk evict+insert, window size $n=\textrm{4,194,304}$, bulk size $m=\textrm{1}$, varying ooo distance $d$.}
\end{figure*}

\textbf{Figure~\ref{fig:bulk_evict_insert_bulksize1}}
shows the throughput as a function of the out-of-order degree~$d$ when
running with neither bulk evict nor bulk insert, i.e.\ \mbox{with $m=1$}.
As before, this experiment elides algorithms that require in-order data.
%% \pseudocode{amta}, \pseudocode{twostacks\_lite}, and
%% \pseudocode{daba} are elided since they require in-order data.
In the absence of bulk operations, we expect 
\pseudocode{b\_fiba} to have no advantage over
\pseudocode{nb\_fiba}.
In practice, \pseudocode{b\_fiba} does worse on \pseudocode{sum} and
\pseudocode{geomean} but slightly better on \pseudocode{bloom}.

\begin{figure*}
\centerline{%
  \includegraphics[width=.32\textwidth]{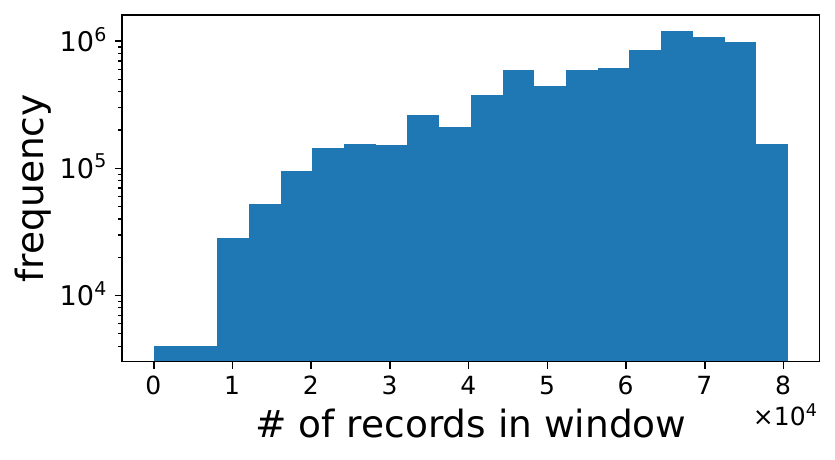}
  \hspace*{3mm}
  \includegraphics[width=.32\textwidth]{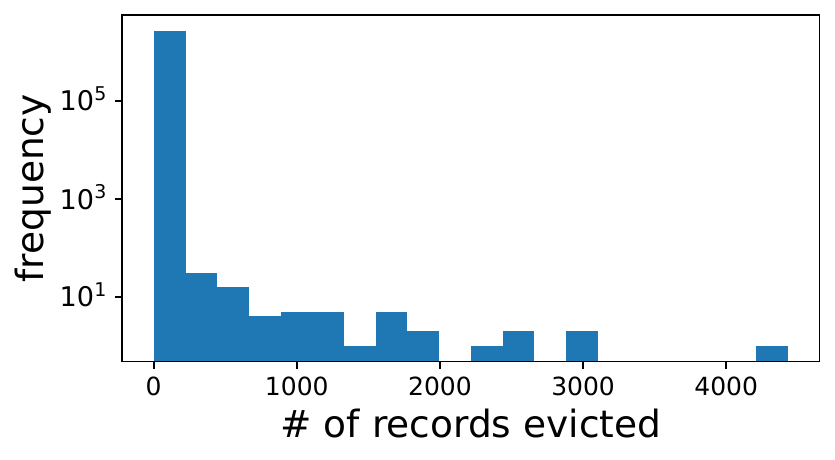}
  \hspace*{3mm}
  \includegraphics[width=.32\textwidth]{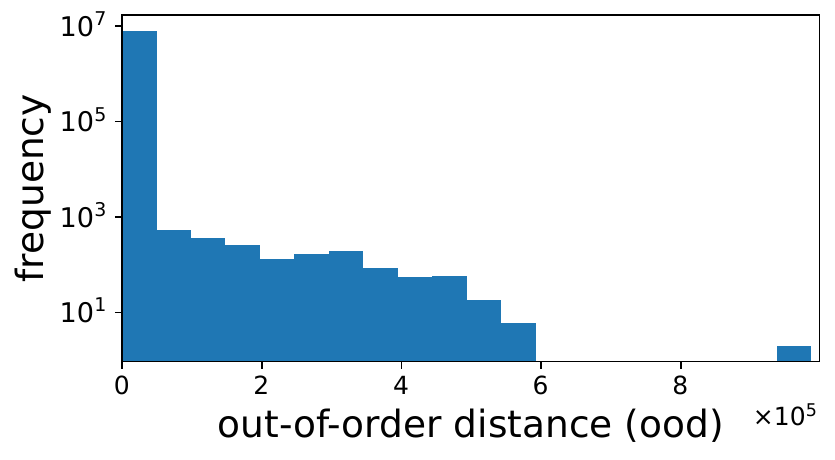}%
}
\vspace*{-4mm}
  \caption{\label{fig:citi_histograms}Histograms of (left) citi bike instantaneous window sizes~$n$, 
  (middle) eviction bulk sizes~$m$ for a time-based window of 1~day, and 
  (right) the out-of-order distance~$d$, i.e., the number of records skipped over by insertions.}
\end{figure*}
\begin{figure*}
\centerline{\includegraphics[width=\textwidth]{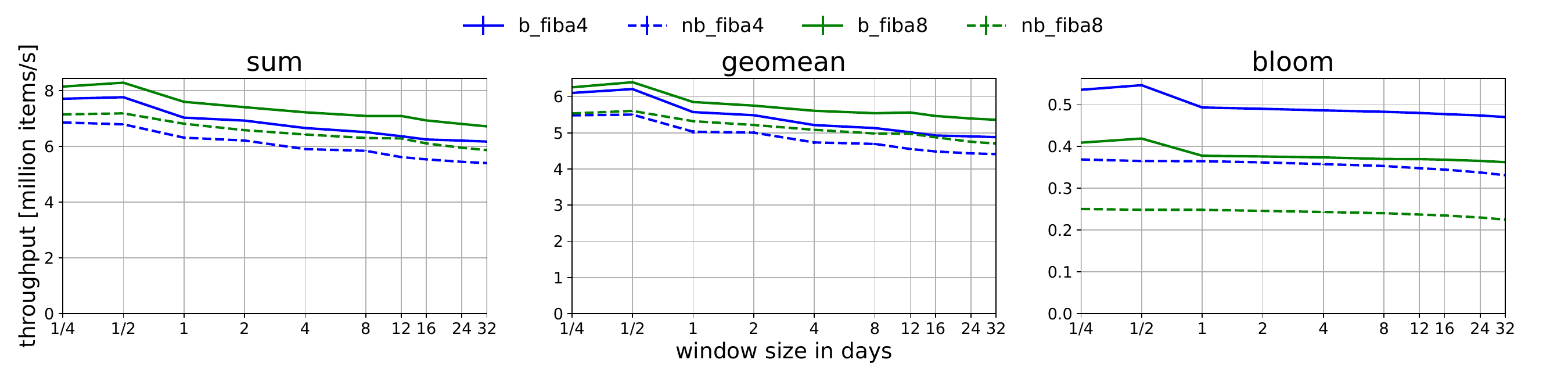}}
\vspace*{-4mm}
\caption{\label{fig:bike}Throughput, citi bike, varying window size $n$, bulk size $m$ and ooo distance $d$ from real data.}
\end{figure*}

\vspace{-2mm}
\subsection{Window Size One Billion}\label{sec:one_billion}
To understand how our algorithm behaves in more extreme scenarios, we ran
\pseudocode{b\string_fiba4} with \pseudocode{geomean}
with a window size of 1~billion ($n = 10^{9}$). 
In theory, FiBA is expected to grow to any window
size and have good cache behaviors, like a B-tree. In practice, this is the case at window
size 1B: The benchmark ran uneventfully using 99\% CPU on average, fully utilizing the one core that it has. Memory
occupancy per window item (i.e., the maximum resident set size for
the process divided by the window size) stays the same ($64-70$
bytes), independent of window size.  
%This small variation is mainly a
%function of how full the tree nodes are.

However, at $n = 1\mbox{B}$, the benchmark has a larger overall memory
footprint, putting more burden on the memory system. This directly manifests as
more frequent cache misses/page faults and indirectly affects the
throughput/latency profile.  While no major page faults
occurred, the number of minor page faults \emph{per} million tuples
processed increased multiple folds (657 at $4\mbox{M}$ vs. 15,287 at
$1\mbox{B}$).
Compared with the 4M-window experiments, the throughput numbers for $n =
1\mbox{B}$ mirror the same trends as the bulk size is varied.  In
absolute numbers, the throughput of $n = 1\mbox{B}$ is $1-1.12\times$ less
than that of $n = 4\mbox{M}$.  For latency, the theory promises 
$\log d$ average (amortized) bulk-evict time, independent of the window size.
With a larger memory footprint, however, we expect a 
slight increase in median latency. The $\log n$ worst-case time should mean the rare spikes
will be noticeably higher with larger window sizes. In practice, we observe
that the median only goes up by $\approx 7.5\%$.  The 99.999-th
percentile markedly increases by around $2\times$.
%% , but such rare events might just happen by chance. 
% the five 9s
%   -- 23966 (cycles) @ 1B
%   -- 11350 (cycles) @ 4M

\subsection{Real Data}\label{sec:results_real}

The previous experiments carefully controlled the variables $n$, $m$,
and~$d$ to explore tradeoffs and validate the theoretical results.
It is also important to see how the algorithms perform on real data.
Specifically, real applications tend to use time-based
windows~(causing both $n$ and $m$ to fluctuate), and real data tends
to be out-of-order~(with varying~$d$).
In other words, all three variables vary within a single run.
\textbf{Figure~\ref{fig:citi_histograms}}
shows this for the NYC Citi Bike dataset~\cite{citi_bike_2022} (Aug--Dec 2018).
The figure shows a histogram of window sizes $n$~(left) and a histogram of
bulk sizes $m$~(middle), assuming a time-based sliding window of 1~day.
Depending on whether that 1~day currently contains more or fewer stream
data items, $n$ ranges broadly, as one would expect for real data
whose event frequencies are uneven.
Similarly, depending on the time\-stamp of the newest inserted window
entry, it can cause a varying number $m$ of the oldest entries to be
evicted.
Most single insertions cause only a single eviction, but there are a
non-negligible number of bulk evicts of hundreds or thousands of
entries.
The figure also shows a histogram of out-of-order distances $d$~(right).
While the vast majority of insertions have a small out-of-order
distance~$d$, there are also hundreds of insertions with $d$ in the
tens of thousands.

\textbf{Figure~\ref{fig:bike}}
shows the throughput results for the Citi Bike dataset on a run that
involves bulk evicts with varying $m$ and single inserts with
varying~$d$.
Since \pseudocode{amta}, \pseudocode{twostacks\_lite}, and
\pseudocode{daba} require in-order data, we cannot use them here.
In theory, we expect the bulk operations to give \pseudocode{b\_fiba}
an advantage over \pseudocode{nb\_fiba}.
In practice, we find that this is indeed the case for real-world data.

\subsection{Java and Apache Flink}\label{sec:results_flink}

\begin{figure}
\centerline{\includegraphics[width=0.32\textwidth]{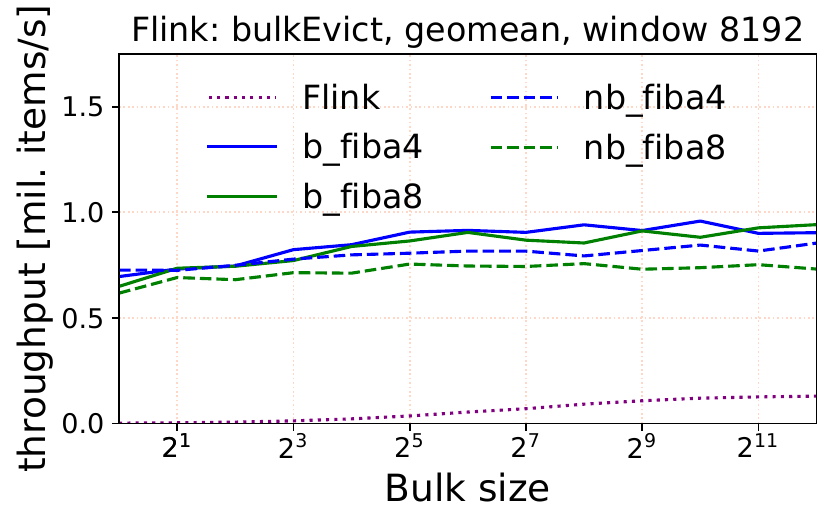}}
\vspace*{-4mm}
\caption{\label{fig:flink_bulkevict_geomean}Throughput, Flink, bulk evict only, window size $n=\textrm{8,192}$, varying bulk size $m$, in-order data $d=\textrm{0}$.}
\vspace*{-4mm}
\end{figure}

To experiment with our algorithm in the context of an end-to-end system, we
reimplemented it in Java inside Apache Flink 1.17~\cite{carbone_et_al_2015}. We
ran experiments that repeatedly perform several single inserts followed by a
bulk evict and query. Using a window of size $n=2^{22} \approx 4\text{M}$, the
FiBA algorithms perform as expected but the Flink baseline was prohibitively
slow, so we report a comparison at $n = 8,192$ instead. At this size, the
trends are already clear. Figure~\ref{fig:flink_bulkevict_geomean} shows that
even without our new bulk eviction support, FiBA is much faster than Flink.
Using bulk evictions further widens that gap. As expected, throughput improves
with increasing bulk size~$m$, consistent with our findings with C++
benchmarks.

\section{Conclusion}\label{sec:conclusion}

This paper describes algorithms for bulk insertions and evictions for
incremental sliding-window aggregation.
Such bulk operations are necessary for real-world data streams, which
tend to be bursty.
Furthermore, real-world data streams tend to have out-of-order data.
Hence, besides handling bulk operations, our algorithms also handle
that case.
Our algorithms are carefully crafted to yield the same algorithmic
complexity as the best prior work for the non-bulk case while
substantially improving over that for the bulk case.\phantom{\cite{tangwongsan_hirzel_schneider_2023}}

%% This paper presents proofs that this is indeed true in theory and
%% empirical experiments showing this is also true in practice.
%% Overall, this paper advances the state of the art for sliding-window
%% aggregation, making it applicable to a broad set of circumstances.

\bibliographystyle{ACM-Reference-Format}
\balance
\bibliography{bibfile}

\ifVldbElse{%
}{%
\newpage
\appendix
\section{Proofs}\label{sec:proofs}
% We flesh out what was deferred from earlier sections.

\begin{proof}[Proof~of~Claim~\ref{claim:well-splittable}]
   Let $k = \left\lfloor \frac{p}{\mu + 1} \right\rfloor$. Then, $p = k(\mu + 1)
   + r$, where $0 \leq r \leq \mu$. There are two cases to consider:
   If $r = \mu$, set $t = k$ and $\beta_t = \mu$. 
    Otherwise, $r \leq \mu - 1$, set $t = k-1$ and
         $\beta_t = \mu + 1 + k \leq 2\mu$---i.e., the remainder is too small
         to be a valid node, so we combine them with the rightmost arity-$(\mu
         + 1)$ node thus far.  Notice that $t \geq 2$ since $p \geq 2\mu + 2$. 
   %\end{itemize}
\end{proof}

\newcommand*{\minarity}{\mu}
%\begin{proof}[Proof~of~Lemma~\ref{lem:fiba-bulk-restructuring}]
\noindent{}\textbf{Further Discussion of Lemma~\ref{lem:fiba-bulk-restructuring}.}
The amortized bound can be analyzed analogously to Lemma~9 in the FiBA
paper~\cite{tangwongsan_hirzel_schneider_2019}, arguing that charging $2$ coins
per new entry is sufficient in maintaining the tree. Below, we describe the
extension from that proof that pertains to \pseudocode{bulkInsert}.

For a node $w$, we expect a coin reserve of 
\[\small \textit{coins}(w)=\left\{\begin{array}{ll}
    2 + 2k & \mbox{if $a=2\minarity+k$ and $k \geq 1$}\\
    2 & \mbox{if $a=2\minarity$ or ($a=\minarity-1$ and $w$ is not the root)}\\
    1 & \mbox{if $a=\minarity$ and $w$ is not the root}\\
    0 & \mbox{if $a<2\minarity$ and ($a>\minarity$ or $w$ is the root)}
  \end{array}\right.,\]
which naturally extends the \textit{coins} function for arity $a \geq 2\mu + 1$.

At the start of \pseudocode{bulkInsert}, we charge each entry $2$ coins,
sufficient at the insertion sites. As restructuring happens, we reason about
each affected node $w$ as follows: Let $t = \left\lfloor a/(\mu + 1)
\right\rfloor$, so $a = t(\mu + 1) + r$, where $0 \leq r \leq \mu$. There
are two cases:
\begin{itemize}[leftmargin=*]
   \item $r = \mu$: Node $w$ yields $t-1$ new arity-$(\mu+1)$ nodes and one
      arity-$\mu$ node, needing $3(t-1) + 4 = 3t + 1$ coins to send up, pay for the
      split, and keep on the last node. But for $t \geq 2$, the reserve on $w$
      has $\rho := 2(t(\mu+1) + \mu - 2\mu + 1) = 2(t-1)\mu + 2t + 2 \geq
      2\frac{t}2\mu + 2t + 2 \geq 4t + 2$ as $\mu \geq 2$. Also, when $t = 1$,
      $\rho = 4 \geq 3t + 1$.
      
   \item $r \leq \mu - 1$: Node $w$ yields $t-2$ new arity-$(\mu+1)$ nodes
      and one arity $(\mu + 1 + r)$ node, needing at most $3(t-2) + 5 = 3t -
      1$ to send up, pay for the split, and keep on the last node (if $r =
      \mu-1$). Now we know $t \geq 2$. For $t \geq 3$, the reserve $\rho :=
      2(t(\mu+1)+r - 2\mu + 1) \geq 2(t-2)\mu + 2t \geq 2\frac{t}4\mu + 2t
      \geq 3t$ since $\mu \geq 2$.  For $t = 2$, $\rho = 6 + 2r \geq 3t - 1$.
\end{itemize}
Either way, the reserve coins can cover the restructuring cost.

\section{Pseudocode and Examples}\label{sec:examples}

Figures \ref{fig:move_batch_example},
\ref{fig:merge_notsibling_example},
and~\ref{fig:make_child_root_example} provide pseudocode and more
concrete examples for some of the core operations of our algorithm.
To reduce clutter, the figures illustrate the tree structure with
only timestamps.
All three of these figures assume $\pseudocode{MIN\_ARITY}=2$, so each
non-leaf node has between 2 and~4 children and each node has between 1
and~3 entries.

\begin{figure*}
  \begin{center}\begin{tabular}{c}\begin{lstlisting}[language=MyPseudoCode]
fun moveBatch(node: Node, neighbor: Node, ancestor: Node, k: Int)
   a $\gets$ max i $\in$ 0, ..., ancestor.arity - 2 if ancestor.getTime(i) < neighbor.getTime(0)
   if node.isLeaf()
      node.pushBackEntry(ancestor.getTime(a), ancestor.getValue(a))
      for i $\in$ 0, ..., k - 2
         node.pushBackEntry(neighbor.getTime(i), neighbor.getValue(i))
   else
      node.pushBack(ancestor.getTime(a), ancestor.getValue(a), neighbor.getChild(0))
      for i $\in$ 0, ..., k - 2
         node.pushBack(neighbor.getTime(i), neighbor.getValue(i), neighbor.getChild(i + 1))
   ancestor.setEntry(a, neighbor.getTime(k - 1), neighbor.getValue(k - 1))
   neighbor.popFront(k)\end{lstlisting}\end{tabular}\end{center}
  \vspace*{1mm}
  \centerline{\textbf{(a) Pseudocode}}
  \vspace*{3mm}
  \centerline{\includegraphics[scale=1.333]{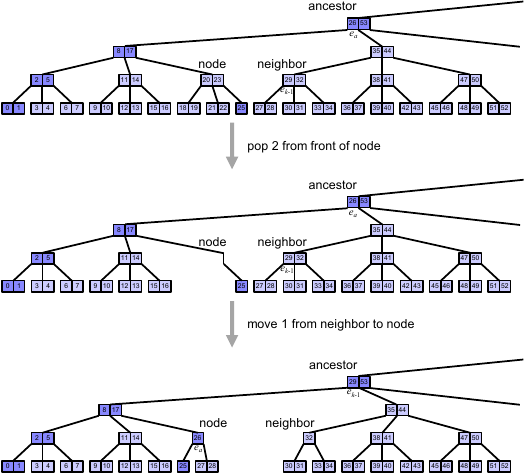}}
  \vspace*{1mm}
  \centerline{\textbf{(b) Concrete example}}
  \caption{\label{fig:move_batch_example}Pseudocode and concrete example
    illustrating moving a batch, c.f.\ Figure~\ref{fig:move_batch_general}.
    \rm
    The example shows a \pseudocode{bulkEvict(24)}, i.e., evicting
    everything up to $t\le24$.
    Specifically, it shows the second iteration of the eviction loop,
    which tackles the next layer of nodes immediately above the leaves.
    The local eviction pops the first two children and entries from
    \pseudocode{node}, leaving a degenerate node remnant with no
    entries and a single child.
    Since \pseudocode{neighbor} has sufficient surplus, to
    re-establish the arity invariant, the algorithm next moves $k=1$
    child and entry to \pseudocode{node}.
    This involves rotating $e_a$ with timestamp~26 from
    \pseudocode{ancestor} to \pseudocode{node};
    moving the first child from \pseudocode{neighbor} to \pseudocode{node};
    and rotating $e_{k-1}$ with timestamp~29 from
    \pseudocode{neighbor} to \pseudocode{ancestor}.
    After this operation, all arities are repaired up to the current level,
    and the next iteration will resume the eviction loop a level higher.}
\end{figure*}

\begin{figure*}
  \begin{center}\begin{tabular}{c}\begin{lstlisting}[language=MyPseudoCode]
fun mergeNotSibling(node: Node, neighbor: Node, ancestor: Node):
   a $\gets$ max i $\in$ 0, ..., ancestor.arity - 2 if ancestor.getTime(i) < neighbor.getTime(0)
   if node.isLeaf()
      neighbor.pushFrontEntry(ancestor.getTime(a), ancestor.getValue(a))
      for i $\in$ node.arity - 2, ..., 0
         neighbor.pushFrontEntry(node.getTime(i), node.getValue(i))
   else
      neighbor.pushFront(node.getChild(node.arity - 1), ancestor.getTime(a), ancestor.getValue(a))
      for i $\in$ node.arity - 2, ..., 0
         neighbor.pushFront(node.getChild(i), node.getTime(i), node.getValue(i))
   ancestor.popFront(a + 1)\end{lstlisting}\end{tabular}\end{center}
  \vspace*{1mm}
  \centerline{\textbf{(a) Pseudocode}}
  \vspace*{3mm}
  \centerline{\includegraphics[scale=1.333]{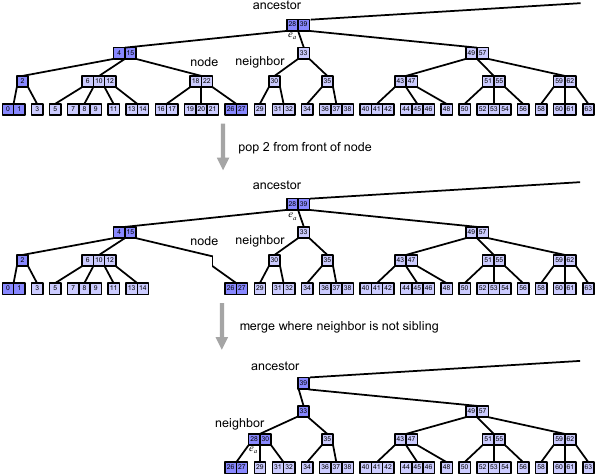}}
  \vspace*{1mm}
  \centerline{\textbf{(b) Concrete example}}
  \caption{\label{fig:merge_notsibling_example}Pseudocode and concrete
    example illustrating merging with a non-sibling neighbor,
    c.f.\ Figure~\ref{fig:merge_notsibling_general}.
    \rm
    The example shows a \pseudocode{bulkEvict(25)}, i.e., evicting
    everything up to $t\le25$.
    Specifically, it shows the second iteration of the eviction loop,
    which tackles the next layer of nodes immediately above the leaves.
    The local eviction pops the first two children and entries from
    \pseudocode{node}, leaving a degenerate node remnant with no
    entries and a single child.
    Since \pseudocode{neighbor} has insufficient surplus, to
    re-establish the arity invariant, the algorithm next merges the
    node into its neighbor.
    Specifically, it moves the sole remaining child from
    \pseudocode{node} to \pseudocode{neighbor}, and it pops the first
    entry $e_a$ with timestamp~28 from \pseudocode{ancestor} and adds
    it to \pseudocode{neighbor}.   
    After this operation, all arities are repaired up to the current
    level, and there is no node for the algorithm to work on at the
    next level, so the eviction loop resumes at the ancestor's layer.}
\end{figure*}

\begin{figure*}
  \begin{center}\begin{tabular}{c}\begin{lstlisting}[language=MyPseudoCode]
if neighbor $=$ $\bot$
   if node.arity $=$ 1 and not node.isLeaf()
      root $\gets$ node.getChild(0)
   else if node $\neq$ root
      root $\gets$ node
   root.becomeRoot()
   repairLeftSpineInfo(root, i $=$ 0)
   top $\gets$ root
   break\end{lstlisting}\end{tabular}\end{center}
  \vspace*{1mm}
  \centerline{\textbf{(a) Pseudocode}}
  \vspace*{3mm}
  \centerline{\includegraphics[scale=1.333]{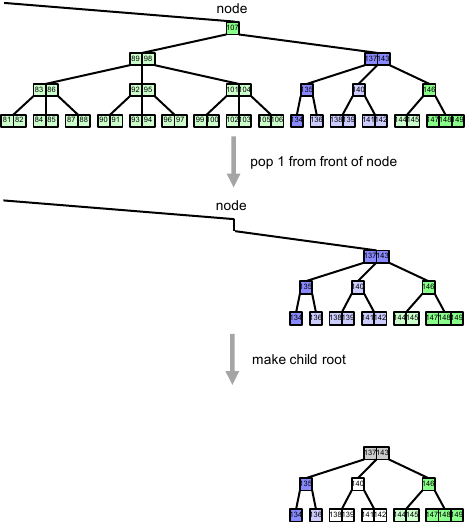}}
  \vspace*{1mm}
  \centerline{\textbf{(b) Concrete example}}
  \caption{\label{fig:make_child_root_example}Pseudocode and concrete example
    illustrating making the child root,
    c.f.\ Figure~\ref{fig:make_child_root_general}.
    \rm
    The example shows a \pseudocode{bulkEvict(107)}, i.e., evicting
    everything up to $t\le107$.
    Specifically, it shows the fourth iteration of the eviction loop,
    which tackles the nodes three layers above the leaves.
    The local eviction pops the first child and entry from
    \pseudocode{node}, leaving a degenerate node remnant with no
    entries and a single child.
    Being on the right spine, \pseudocode{node} has no neighbor that
    is even further right for doing a move or a merge operation.
    That means that the tree must shrink from the top, as everything
    above the node has lower timestamps.
    Since the node has zero entries and a single child, it cannot
    itself become the new root.
    Instead, the single child becomes the root.}
\end{figure*}

}
%\clearpage

\end{document}
\endinput